\documentclass[a4paper,onecolumn,11pt,accepted=2021-08-28]{quantumarticle}
\pdfoutput=1

\usepackage[margin=1in]{geometry}
\usepackage[colorlinks = true, linkcolor = blue, urlcolor  = blue, citecolor = red]{hyperref}
\usepackage{amsmath}
\usepackage{amssymb}  
\usepackage{amsthm}
\usepackage{amsfonts}
\usepackage{graphicx}
\usepackage{bbm}
\usepackage{url}
\usepackage{dsfont}
\usepackage[dvipsnames]{xcolor}

\definecolor{darkgreen}{RGB}{0,100,0}

\newcommand{\ket}[1]{| #1 \rangle}
\newcommand{\bra}[1]{\langle #1 |}
\newcommand{\braket}[2]{\langle #1 | #2 \rangle}
\newcommand{\cB}{\mathcal{B}}
\newcommand{\cH}{\mathcal{H}}
 
\newcommand\N{\mathbb{N}}
\newcommand\M{\mathcal{M}}
\newcommand{\ketbra}[2]{| #1 \rangle \langle #2 |}
\renewcommand{\phi}{\varphi}
\newcommand\C{\mathbb{C}}
\newcommand\R{\mathbb{R}}

\newcommand{\norm}[2]{\ensuremath{|\!|#1|\!|_{#2}}}

\def\beq{\begin{equation}}
\def\eeq{\end{equation}}
\def\bq{\begin{quote}}
\def\eq{\end{quote}}
\def\ben{\begin{enumerate}}
\def\een{\end{enumerate}}
\def\bit{\begin{itemize}}
\def\eit{\end{itemize}}

\def\ra{\rightarrow}

\def\lb{\left(}
\def\rb{\right)}
\def\lset{\lbrace}
\def\rset{\rbrace}

\def\r|{\right|}
\def\lbr{\left[}
\def\rbr{\right]}

\newtheorem{theorem}{Theorem}[section]
\newtheorem{lemma}[theorem]{Lemma}

\newtheorem{definition}[theorem]{Definition}

\theoremstyle{remark}
\newtheorem{example}[theorem]{Example}
\newtheorem{remark}[theorem]{Remark}

\begin{document}

\title{A refinement of Reznick's Positivstellensatz with applications to quantum information theory}

\author{Alexander M\"uller-Hermes}
\email{muellerh@posteo.net}
\affiliation{Department of Mathematical Sciences, University of Copenhagen, 2100 Copenhagen, Denmark \hfill\\ Institut Camille Jordan, Universit\'{e} Claude Bernard Lyon 1, 43 boulevard du 11 novembre 1918,\\ 69622 Villeurbanne cedex, France}
\author{Ion Nechita}
\email{ion.nechita@univ-tlse3.fr}
\affiliation{Laboratoire de Physique Th\'eorique, Universit\'e de Toulouse, CNRS, UPS, France}
\author{David Reeb}
\email{reeb.qit@gmail.com}
\affiliation{Institute for Theoretical Physics, Leibniz Universit\"at Hannover, 30167 Hannover, Germany \hfill\\ Bosch Center for Artificial Intelligence, Robert-Bosch-Campus 1, 71272 Renningen, Germany}
\maketitle

\begin{abstract}
In his solution of Hilbert's 17th problem Artin showed that any positive definite polynomial in several variables can be written as the quotient of two sums of squares. Later Reznick showed that the denominator in Artin's result can always be chosen as an $N$-th power of the squared norm of the variables and gave explicit bounds on $N$. By using concepts from quantum information theory (such as partial traces, optimal cloning maps, and an identity due to Chiribella) we give simpler proofs and minor improvements of both real and complex versions of this result. Moreover, we discuss constructions of Hilbert identities using Gaussian integrals and we review an elementary method to construct complex spherical designs. Finally, we apply our results to give improved bounds for exponential quantum de Finetti theorems in the real and in the complex setting.
\end{abstract}

\tableofcontents

\section{Introduction}
In the same way that the Nullstellensatz is fundamental for complex algebraic geometry, so called Positivstellens\"atze are important in real algebraic geometry \cite{bochnak2013real,marshall2008positive}. A Positivstellensatz \cite{krivine1964anneaux,stengle1974nullstellensatz} states that a polynomial in $d$ real variables which is non-negative on some subset of $\R^d$ is related in some prescribed way to a \emph{sum of squares} (SOS), which are special polynomials guaranteed by definition to be non-negative. Most of such results consider polynomials which are non-negative on \emph{semialgebraic sets} (sets where a finite number of polynomials are non-negative) and other need a (strict) positivity guarantee (e.g.~Schm\"udgen's \cite{schmuedgen1991positivstellensatz} and Putinar's \cite{putinar1993positive} Positivstellens\"atze). In this work, we shall focus on results close to Artin's solution to Hilbert's 17th problem \cite{artin1927zerlegung}: 

\medskip

\emph{For any strictly positive homogeneous polynomial $p$ in $d$ real variables, there exist two sum-of-squares polynomials $h,q$ such that $hp = q$.}

\medskip

In his seminal work \cite{reznick1995uniform}, Reznick showed that $h$ can be taken of the form $h(x) = \|x\|^{2N}$, giving also bounds on $N$, in terms of the number of variables, the degree, and a certain measure of positivity of $p$. We re-prove this type of results, both in the real \cite{reznick1995uniform} and in the complex \cite{to2006effective} cases, using techniques from quantum information theory. 

The tools from quantum information theory we employ are related mainly to the entanglement theory of symmetric, multi-partite quantum states. A great introduction to the main ideas and techniques we deploy is Harrow's preprint \cite{harrow2013church}. We also develop the parallel theory in the real case, which is less known than the complex variable case. Our main technical insight is an explicit inversion of a well-known identity due to Chiribella \cite{chiribella2010quantum} relating three sequences of quantum maps: the measure-and-prepare maps, the partial traces, and the approximate cloning maps. 

The main contribution of this work is to make precise the deep connection between Reznick-type Positivstellens\"atze and quantum information theory by recasting the classical proofs of the former in the linear algebraic language of the latter. As a byproduct of our careful analysis of this correspondence, we slightly improve the bounds on the exponents needed in the Positivstellens\"atze and in exponential de Finetti theorems, following \cite{harrow2013church}. 

When finishing our article we learned of the recent work by Fang and Fawzi \cite{fang2020sum} improving the convergence rates of sums-of-squares hierarchies by polynomial techniques related to Reznick's ideas (see also \cite{faybusovich2004global,doherty2012convergence} for other papers analyzing the speed of convergence of SDP hierarchies for polynomial optimization). While our work is also based on these ideas, our focus is quite different. Instead of estimating when a polynomial is a sum-of-squares we are interested in the particular form of the decomposition that is central in Reznick's work. However, it would be interesting to see if the techniques of Fang and Fawzi could also lead to new results in this direction. We shall keep this question for future study. 

Our paper is organized as follows. In Section \ref{sec:preliminaries} we introduce the correspondence between bi-Hermitian homogeneous multi-variable polynomials and Hermitian operators acting on the symmetric subspace of a tensor power, emphasizing the direct correspondence between analytical and algebraic operations. Sections \ref{sec:PSS-C} and \ref{sec:PSS-R} contain the proofs of the complex, resp.~real Positivstellens\"atze. In Section \ref{sec:exp-deF} we discuss exponential de Finetti theorems. The Appendices contain results on Hilbert identities and complex spherical designs used in the proofs. 

\section{Preliminaries}\label{sec:preliminaries}

In this section we set the stage for the proof of our main result, the complex Positivstellensatz in Theorem \ref{thm:PSS-C}. We do so by discussing the folklore connection between bi-Hermitian forms and Hermitian operators acting on the symmetric subspace of a tensor power of a finite dimensional complex Hilbert space. We then relate various linear algebraic operations on these operators to natural maps on the corresponding polynomials. We equally discuss the only purely analytical tool used in this paper to establish both the complex and the real Positivstellens\"atze, the Bernstein inequality in Lemmas \ref{lem:Bernstein-real} and \ref{lem:Bernstein}. 

We shall denote by $\vee^n\C^d$ the \emph{symmetric subspace} of the tensor product $(\C^d)^{\otimes n}$, and by $\cB(\vee^n\C^d)$ and $H(\vee^n\C^d)$ the spaces of (bounded) linear operators and Hermitian operators respectively from $\vee^n\C^d$ to itself. The space $\vee^n\C^d$ is spanned by the family $\{x^{\otimes n}\, : \, x \in \C^d \}$, see \cite[Section I.5]{bhatia1997matrix} or \cite[Theorem 3]{harrow2013church}. Importantly, we denote by $d[n]$ the dimension of the symmetric subspace:
$$d[n] := \dim \vee^n\C^d = \binom{d+n-1}{n}.$$

We shall use $\mathcal S^{d-1}$ to denote the complex unit sphere of $\C^d$, and $\mathfrak S_n$ to denote the permutation group on $n$ elements. We shall also use the falling factorial notation 
$$(x)_p := x(x-1) \cdots (x-p+1)$$
for real $x$ and integer $p \geq 1$. We use the bra-ket notation from quantum mechanics, denoting e.g.~by $\ketbra{x}{y}$ the rank-one linear operator $xy^*$.

\subsection{Polynomials and operators acting on the symmetric subspace}

For any Hermitian operator $W=W^*\in\mathcal H(\vee^{k_1}\C^{d_1}\otimes\ldots\otimes\vee^{k_l}\C^{d_l})$ we consider the corresponding bi-Hermitian form in the complex variables $x_1\in\C^{d_1}$, \ldots, $x_l\in\C^{d_l}$,
\[
p_W(x_1,\ldots,x_l) := \bra{x_1^{\otimes k_1}\otimes\ldots\otimes x_l^{\otimes k_l}}W\ket{x_1^{\otimes k_1}\otimes\ldots\otimes x_l^{\otimes k_l}}.
\]
The terminology "bi-Hermitian" used above refers to the fact that the form $p_W$, taking as input $l$ vectors, is of homogeneous degree $k_i$ in $x_i$ and $\overline{x_i}$; moreover, it determines the operator $W$ uniquely. Therefore, we shall often switch between the ``operator picture'' involving $W$ and the equivalent ``polynomial picture'' involving $p_W$. We introduce the following notation for the extremal values of $p_W$ on the unit sphere of each of the $l$ variable sets:
\begin{align*}
m(W) &:= \min_{\forall i:\|x_i\|_2 = 1}  p_W(x_1,\ldots,x_l), \\
M(W) &:= \max_{\forall i:\|x_i\|_2 = 1}  p_W(x_1,\ldots,x_l).
\end{align*}
An important special case we shall consider is the case $l=1$, in which we often write $x_1\equiv x$, $d_1\equiv d$, $k_1\equiv k$, and $p_W$ is called the \emph{Q function} \cite{husimi1940some}. In this case, $W$ is called \emph{block-positive} iff $m(W)\geq0$, that is if the corresponding polynomial has non-negative range. As a more general case we consider $l=2$ and $k_2=1$, i.e.~where the variables $x_2$ do not appear in the special tensor-product structure in $p_W$. In this case, we furthermore denote $x_2\equiv y$, and $d_2\equiv D$ (see Section \ref{sec:PSS-C}); of course, this reduces to the previous case upon setting $D=1$, $y=1$.

For $n\geq k$ we denote by $\operatorname{tr}_{n\ra k}:\cB\lb\vee^n \C^d\rb\ra \cB\lb\vee^k \C^d\rb$ the partial trace erasing $n-k$ systems. In the polynomial picture the partial trace reduces to a differential operator given in terms of the  Laplacian 
\begin{equation}\label{eq:complex-Laplacian}
\Delta = \sum_{i=1}^d \frac{\partial^2 p}{\partial \bar x_i \partial x_i},
\end{equation}
where we formally treat $x_i$ and $\bar{x}_i$ for $i\in \{ 1,\ldots ,d\}$ as independent variables. 

\begin{lemma}
For any $W \in \mathcal H(\vee^k \mathbb C^d)$, we have
\[
p_{\operatorname{tr}_{k \to k-t}(W)} = ((k)_t)^{-2} \Delta^t p_W.
\]
\label{lem:traceAsDiff}
\end{lemma}

\begin{proof}
Recall that the set $\lset \ketbra{v}{v}^{\otimes k}: \ket{v}\in \C^d\rset$ spans $\mathcal H\lb\vee^{k}\C^d\rb$ (see \cite[Eq. 11b]{harrow2013church}). Therefore, it will be sufficient to show the lemma for the corresponding set of bi-Hermitian forms $\lset p_{\ketbra{v}{v}^{\otimes k}}(x) = |\braket{x}{v}|^{2k}\rset$. Note that on one hand
\begin{align*}
\Delta p_{\ketbra{v}{v}^{\otimes k}}(x) &= \Delta \lbr(v_1\bar{x}_1 + \cdots +v_d\bar{x}_d)^{k}(\bar{v}_1x_1 + \cdots +\bar{v}_dx_d)^{k}\rbr\\
 &= k^2\norm{v}{}^2|\braket{x}{v}|^{2k-2} = k^2\norm{v}{}^2 p_{\ketbra{v}{v}^{\otimes (k-2)}}(x).
\end{align*}
On the other hand 
\[
\text{tr}_{k\ra (k-1)}(\ketbra{v}{v}^{\otimes k}) = \norm{v}{}^2\ketbra{v}{v}^{\otimes (k-2)}.
\]
Direct comparison of the two expressions shows that
\[
k^2\text{tr}_{k\ra (k-1)}(p) = \Delta p .
\] 
Finally, by iterating the previous formula the statement of the lemma follows. 
\end{proof}

\begin{figure}
	\centering
\begin{minipage}{6in}
	\centering
	\raisebox{-0.5\height}{\includegraphics[width=0.3\linewidth]{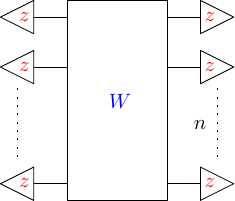}}
	\hspace*{.2in}
	\raisebox{-0.5\height}{\includegraphics[width=0.3\linewidth]{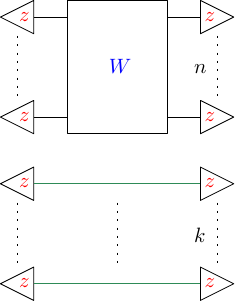}}
	\hspace*{.2in}
	\raisebox{-0.5\height}{\includegraphics[width=0.3\linewidth]{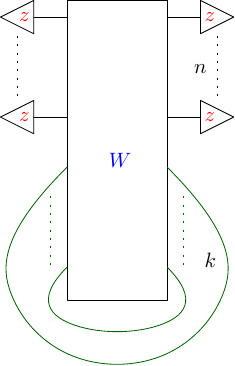}}	
\end{minipage}	
	\caption{Graphical representation of the correspondence between self-adjoint operators $\textcolor{blue}W$ acting on the symmetric subspace, and polynomials. From left to right, we have depicted the diagrams for $p_{\textcolor{blue}{W}}(\textcolor{red}{z})$, $\textcolor{darkgreen}{\|\textcolor{red}{z}\|^{2k}}p_{\textcolor{blue}{W}}(\textcolor{red}{z})$, and $\textcolor{darkgreen}{((n+k)_k)^{-2} \Delta^k} p_{\textcolor{blue}{W}}(\textcolor{red}{z})$ respectively, where $\Delta$ is the complex Laplacian \eqref{eq:complex-Laplacian}. This emphasized in particular that multiplying with the norm and the iterated complex Laplacian are, up to constants, dual operations.}
	\label{fig:diagrams-polynomials-complex}
\end{figure}

\subsection{Spherical designs}

In order to have discrete versions of our main result, the complex Positivstellensatz from Theorem \ref{thm:PSS-C}, we need the following relaxation of the uniform measure on the complex unit sphere. The real case has a long history in mathematics and computer science \cite{delsarte1991spherical}, while the complex case has received a lot of interest due to applications in quantum information theory \cite{scott2006tight}.

\begin{definition}\label{definition-t-design}
For any $n,d\in\mathbb N$, a \emph{complex spherical $n$-design} (in $d$ dimensions) is a measure $\mathrm d\phi$ on the complex unit sphere ${\mathcal S}^{d-1}:=\{\phi\in{\mathbb C}^d:|\phi|=1\}$ such that
\begin{equation*}
d[n]\int_{\mathcal S^{d-1}} |\phi\rangle\langle\phi|^{\otimes n} \mathrm d\phi=P_{sym}^{(n,d)},
\end{equation*}
where $P_{sym}^{(n,d)}$ is the orthogonal projection on the symmetric subspace $\vee^n\mathbb C^d\subseteq(\mathbb C^d)^{\otimes n}$.
\end{definition}

Integration with respect to a spherical $n$-design over a polynomial of degree at most $n$ in $\phi \in \C^d$ and degree at most $n$ in $\overline{\phi}$ therefore yields the same result as integration with respect to the Haar measure (the unique unitarily invariant probability measure on ${\mathcal S}^{d-1}$), which is a spherical $n$-design for any $n\in\mathbb N$. But whereas the Haar measure is non-atomic, there exist, for any $n<\infty$, \emph{discrete} spherical $n$-designs supported on a finite number of points, so that integrals turn into finite (weighted) sums; in Appendix \ref{spherical-design-appendix} we show how to construct a complex spherical $n$-design supported on $(n+1)^{2d}$ points.

\subsection{The measure-and-prepare map}

The term \emph{measure-and-prepare map} comes from quantum information theory, where linear maps of a similar form are known as \emph{quantum-classical channels}, see \cite[Sec.~4.6.6]{wilde2017quantum}. Physically, they can be seen as processes where the input is measured in some (possibly over-complete) basis, and then a specific output is prepared (cf. \eqref{equ:MP}). 

\begin{definition}
For $n,k\in\N$, the \emph{measure-and-prepare map} $\operatorname{MP}_{n\to k}:\cB(\vee^n \C^d)\to \cB(\vee^k \C^d)$ is defined as 
\begin{equation}\label{equ:MP}
\operatorname{MP}_{n\to k}(X) := d[n+k]\int \bra{\phi^{\otimes n}} X\ket{\phi^{\otimes n}} \ketbra{\phi}{\phi}^{\otimes k} \mathrm{d}\phi,
\end{equation}
for any $X\in\cB(\vee^n \C^d )$. Here $d[n+k]$ denotes the dimension of the symmetric subspace $\vee^{n+k}\C^d$ and $\mathrm{d}\phi$ denotes the Haar measure on the unit sphere in $\mathbb C^d$ (or any spherical $(n+k)$-design, see Definition \ref{definition-t-design}). 
\end{definition}

Note that the measure-and-prepare map is completely positive, but in general it is neither trace-preserving nor unital. To make it trace-preserving one has to multiply with the scalar $d[n]/d[n+k]$. For any $n,k\in\N$ the measure-and-prepare map has the adjoint $\operatorname{MP}^*_{n\to k} = \operatorname{MP}_{k\to n}$ with respect to the Hilbert-Schmidt inner product. 

The adjoint of the partial trace with respect to the Hilbert-Schmidt inner product on $\cB(\vee^n \C^d )$ is given by 
\begin{equation}\label{eq:adjoint-of-trace}
\operatorname{tr}^{*}_{k\to n}(X) = (\operatorname{tr}_{n\to k})^{*}(X) = P^{(n,d)}_{sym} ( X\otimes I^{\otimes (n-k)}_{d}) P^{(n,d)}_{sym}.
\end{equation}
On the level of polynomials we have 
\begin{equation}\label{eq:adjoint-of-trace-poly-language}
p_{\operatorname{tr}^{*}_{k\to n}(X)}(x) = \langle x^{\otimes n} | \operatorname{tr}^{*}_{k\to n}(X) | x^{\otimes n} \rangle = \langle x^{\otimes n} | X \otimes I_d^{\otimes (n-k)} | x^{\otimes n} \rangle = \|x\|^{2(n-k)} p_X(x).
\end{equation}
The adjoint of the partial trace map is equal, up to a factor, to the so-called ``cloning channel'' 
$$\operatorname{Clone}_{k\to n}:=\frac{d[k]}{d[n]}\operatorname{tr}^*_{k\to n},$$
which is the best quantum-channel approximation to a quantum cloner, mapping $k$ copies of a quantum state to $n$ (approximate) copies, see \cite{keyl1999optimal}.

The measurement-and-prepare map satisfies the following remarkable identity involving partial traces and their adjoints, due to Chiribella \cite[Eq.~(6)]{chiribella2010quantum} (see also \cite[Theorem 7]{harrow2013church}):

\begin{theorem}[Chiribella identity~\cite{chiribella2010quantum}]\label{thm:Chiribella}
For any $n,k\in \N$ we have 
\begin{align}
\label{eq:Chiribella-formula-C}
\operatorname{MP}_{n\to k} &= \sum^{\min(n,k)}_{s=0} c(n,k,s) \operatorname{tr}^{*}_{s\to k}\circ \operatorname{tr}_{n\to s}\\
\nonumber &= \sum^{\min(n,k)}_{s=0} c(n,k,s) \frac{d[k]}{d[s]} \operatorname{Clone}_{s\to k}\circ \operatorname{tr}_{n\to s},
\end{align}
where 
$$c(n,k,s) = \frac{\binom{k}{s}\binom{n}{s}}{\binom{k+n}{k}}.$$
Note that $c(n,k,s)=c(k,n,s)$ and $\sum^{\min(n,k)}_{s=0} c(n,k,s) = 1$.
\label{thm:chiribella}
\end{theorem} 
For the sake of completeness we give the proof of the Chiribella identity presented in \cite[Theorem 7]{harrow2013church}.

\begin{proof}
 For any $a,b\in \C^d$ we have (the integral is, as usual, on the unit sphere of $\C^d$, and $\mathrm{d}\phi$ is a $(n+k)$-spherical design):
\begin{align*}
\bra{b^{\otimes k}}\operatorname{MP}_{n\to k}(\ketbra{a}{a}^{\otimes n})\ket{b^{\otimes k}} &= d[n+k]\int \braket{b^{\otimes k}\otimes a^{\otimes n}}{\phi^{\otimes (n+k)}} \braket{\phi^{\otimes (n+k)}}{b^{\otimes k}\otimes a^{\otimes n}}\text{d}\phi\\
&= \bra{b^{\otimes k}\otimes a^{\otimes n}}P^{(n+k,d)}_{sym} \ket{b^{\otimes k}\otimes a^{\otimes n}} \\
&= \frac{1}{(n+k)!}\sum_{\sigma\in \mathfrak S_{n+k}} \bra{b^{\otimes k}\otimes a^{\otimes n}}P_{\sigma} \ket{b^{\otimes k}\otimes a^{\otimes n}} \\
&= \sum^{\min(n,k)}_{s=0} c(n,k,s)\|a\|^{2(n-s)}\|b\|^{2(n-k)}|\braket{a}{b}|^{2s}\\
&= \left\langle b^{\otimes k} \left | \sum^{\min(n,k)}_{s=0} c(n,k,s) \operatorname{tr}^*_{s\to k}\circ \operatorname{tr}_{n\to s}(\ketbra{a}{a}^{\otimes n}) \right| b^{\otimes k} \right\rangle
\end{align*}
Above, we have used Lemma \ref{lem:SphHilbCompLA} for the second equality and the definition of the projector onto the symmetric subspace as a sum of tensor-permutation matrices for the third equality. To see the fourth equality, note that, among the permutations $\sigma\in \mathfrak S_{n+k}$, precisely 
\[
n!~k!\binom{k}{s}\binom{n}{s}= (n+k)!c(n,k,s)
\]
of them yield $\bra{b^{\otimes k}\otimes a^{\otimes n}}P_{\sigma} \ket{b^{\otimes k}\otimes a^{\otimes n}} = \|a\|^{2(n-s)}\|b\|^{2(n-k)}|\braket{a}{b}|^{2s}$. The theorem then follows from the fact that the set $\{ x^{\otimes n} \, : \, x \in \C^d\}$ spans $\vee^{n}\C^d$.

Finally, the normalization condition $\sum^{\min(n,k)}_{s=0} c(n,k,s) = 1$ is the well-known \emph{Vandermonde identity} \cite[Eq.~(5.22)]{graham1989concrete} given by
$$\sum^{\min(n,k)}_{s=0} \binom{k}{s}\binom{n}{s} =  \sum^{\min(n,k)}_{s=0} \binom{k}{k-s}\binom{n}{s} = \binom{k+n}{k}.$$
\end{proof}

Let us point out that Theorem \ref{thm:Chiribella} will play a central role in our approach to proving real and complex Positivstellens\"atze; the corresponding step in the original proofs from \cite{reznick1995uniform} and, respectively, \cite{to2006effective}, is played by \emph{Hobson's identity} \cite{hobson1931theory}.

\subsection{Bernstein inequalities}

The last ingredient we need is a Bernstein-type inequality, relating the supremum of the Laplacian of some homogeneous polynomial to the supremum of the polynomial itself. Let us first recall the result in the real case (and, for convenience of our notation, only for polynomials of even degree $2k$).

\begin{lemma}[Bernstein-type inequality, real case;  \cite{reznick1995uniform}]
For any $W\in \cH(\vee^k \R^d)$ we have
\[
\Big{|}(\Delta_\R^{t}p_W)(x)\Big{|}\leq d^t (2k)_{2t} M(W)
\]
whenever $\|x\|_2 = 1$, where $p_W(x): = \bra{x^{\otimes k}}W\ket{x^{\otimes k}}$ is a homogeneous polynomial of degree $2k$ in the real variables $x\in\mathbb R^d$, and $\Delta_\R$ denotes the Laplacian with respect to these $d$ real variables.
\label{lem:Bernstein-real}
\end{lemma}

We shall need later the following complex version of this result. 

\begin{lemma}[Bernstein-type inequality, complex case]
For any $W\in \cH(\vee^k \C^d)$ we have
\begin{equation}\label{eq:Bernstein-complex}
\Big{|}(\Delta^{t}p_W)(x)\Big{|}\leq  (d/2)^t (2k)_{2t} M(W)
\end{equation}
whenever $\|x\|_2 = 1$.
\label{lem:Bernstein}
\end{lemma}
\begin{proof}
The proof is based on a reduction to the real case, and the fact that the ``complex Laplacian'' can be expressed in terms of a real Laplacian (depending on real and imaginary parts), as follows. Consider a polynomial $q = q(z,\bar z) = \sum_{s,t \geq 0} q_{st} z^s \bar z^t$. Its (complex) Laplacian reads 
$$\Delta q\equiv\Delta_\C q = \sum_{s,t>0} st q_{st} z^{s-1}\bar z^{t-1}.$$
Writing now $z=a+ib$, with $a,b \in \R$, and taking the ``real Laplacian'' of $q$ with respect to $a,b$, we have
$$\Delta_\R q  := \left( \frac{\partial^2 }{\partial a^2} + \frac{\partial^2 }{\partial b^2}\right) q.$$
Taking the partial derivatives in $q(a+ib, a-ib)$, we obtain
$$\Delta_\C q = \frac 1 4 \Delta_\R q,$$
a relation which extends trivially to several complex variables. Going back to our complex polynomial $p$, we see it as a homogeneous polynomial of degree $2k$ in $2d$ real variables. Applying Lemma \ref{lem:Bernstein-real}, we obtain
$$\Big{|}(\Delta^{t}p_W)(y)\Big{|}\leq 4^{-t} (2d)^t (2k)_{2t} M(W).$$
\end{proof}

It would be interesting to obtain a tighter Bernstein inequality in the complex case, without using the real Bernstein inequality.

\section{A Positivstellensatz for complex Hermitian bi-homogeneous polynomials}\label{sec:PSS-C}

The following theorem is the main contribution of our paper.

\begin{theorem}\label{thm:PSS-C}
For some $k,d,D\in\N$ consider a Hermitian operator $W\in \mathcal{H}(\vee^k \C^d \otimes \C^D)$ with $m(W) >0$. Then, for any positive integer $n \geq k$ such that
\begin{equation}\label{eq:PSS-bound-n-other}
n \geq \frac{dk(2k-1)}{\ln \left( 1 + \frac{m(W)}{M(W)} \right)} -d-k+1,
\end{equation}
we have, for all $x \in \C^d$ and $y \in \C^D$,
\begin{equation}\label{eq:sos-integral-repn}
\|x\|^{2(n-k)}p_W(x,y) = \int_{\mathcal S^{d-1}}  p_{\tilde W}(\phi,y) |\braket{\phi}{x}|^{2n} \mathrm d\phi
\end{equation}
where $p_{\tilde W}(\phi,y)$ is a bi-homogeneous Hermitian form of bi-degree $k$ in $\phi, \bar \phi$ and bi-degree $1$ in $y, \bar y$, satisfying $p_{\tilde W}(\phi,y) \geq 0$ for all $\phi \in \C^d$ and $y \in \C^D$, and \emph{explicitly computable} in terms of $W$. Here, the measure $\mathrm d\phi$ can be any $(n+k)$-design (see Definition \ref{definition-t-design}) showing that $\|x\|^{2(n-k)}p_W(x,y)$ is a \emph{sum of squares}. In the case $k=1$, the bound \eqref{eq:PSS-bound-n-other} can be improved to
\begin{equation}\label{eq:PSS-bound-n-k1}
n \geq d \frac{M(W)}{m(W)} -d.
\end{equation}
\end{theorem}

Note that our main theorem covers a more general case than the one in \cite{reznick1995uniform}: the polynomials we consider have a set of extra $D$ variables, in the spirit of Quillen's result from \cite{quillen1968representation}; we refer the reader to the Introduction for historical considerations. 

Before we prove Theorem \ref{thm:PSS-C}, let us introduce one of the main technical ingredients we shall use. On $\cB(\vee^k \C^d)$, we define the linear map
\[\Phi^{(n)}_{k \to k} :=   \sum_{s=0}^k c(n,k,s) \operatorname{tr}^*_{s \to k} \circ \operatorname{tr}_{k \to s}.\]
By the Chiribella identity (see Theorem \ref{thm:chiribella}), this map is closely related to the measurement-and-prepare map $\operatorname{MP}_{n \to k}$ introduced in \eqref{equ:MP}. In fact we have, for $n\geq k$,
\begin{equation}
\operatorname{MP}_{n \to k} = \Phi^{(n)}_{k \to k} \circ \operatorname{tr}_{n \to k}.
\label{equ:MPviaPhi}
\end{equation}

One of our main technical observations is that the map $\Phi^{(n)}_{k\to k}$ has a particularly nice, explicit, compositional inverse: 
 
\begin{lemma}
\label{lem:inverse}
For $n\geq k$ we define on $\cB(\vee^k \C^d)$ the linear map
\begin{align*}
\Psi^{(n)}_{k \to k} &:=   \sum_{t=0}^k q(n,k,t)  \operatorname{tr}^*_{t \to k} \circ \operatorname{tr}_{k \to t}\\
 &=\sum_{t=0}^k q(n,k,t) \frac{d[k]}{d[t]}  \operatorname{Clone}_{t \to k} \circ \operatorname{tr}_{k \to t}
\end{align*}
with 
\begin{equation}\label{eq:def-q}
q(n,k,t) :=  (-1)^{t+k}\frac{\binom{n+t}{t}\binom{k}{t}}{\binom{n}{k}}\frac{d[n+t]}{d[n+k]}.
\end{equation}
Then, we have on $\cB(\vee^k \C^d)$:
\[
\Phi^{(n)}_{k \to k} \circ \Psi^{(n)}_{k \to k}  = \operatorname{id}_{k \to k}.
\]
\end{lemma}
 \begin{proof}
Since the map $\operatorname{tr}^*_{k \to n}$ is injective (as, for $n\geq k$, its adjoint $\operatorname{tr}_{n\to k}$ is surjective) and the map $\Phi^{(n)}_{k \to k}$ is selfadjoint, the claim is equivalent to showing 
$$\operatorname{tr}^*_{k \to n} = \operatorname{tr}^*_{k \to n} \circ \Phi^{(n)}_{k \to k} \circ \Psi^{(n)}_{k \to k} = (\operatorname{MP}_{n \to k})^* \circ \Psi^{(n)}_{k \to k},$$
or, by taking adjoints, to the following equality of linear maps
\begin{align}
\nonumber \operatorname{tr}_{n \to k} &= \Psi^{(n)}_{k \to k} \circ  \operatorname{MP}_{n \to k} \\
\nonumber &=\sum_{t=0}^k q(n,k,t) \operatorname{tr}^*_{t \to k}\circ  \operatorname{tr}_{k \to t} \circ \operatorname{MP}_{n \to k}\\
\label{eq:inverse-main-equality}&= \sum_{t=0}^k q(n,k,t) \operatorname{tr}^*_{t \to k}\circ \frac{d[n+k]}{d[n+t]}\operatorname{MP}_{n \to t},
\end{align}
on which we are focusing next. 
We use the same idea from the proof of Theorem \ref{thm:Chiribella}: the equality above holds if and only if, when applying the maps to the element $|\alpha \rangle \langle \alpha |^{\otimes n}$ and taking the scalar product with $|\beta \rangle \langle \beta |^{\otimes k}$, we obtain identical results, for all unit norm $\alpha, \beta \in \mathbb C^d$. Letting $x = |\langle \alpha, \beta \rangle|^2 \in [0,1]$, we obtain, for the left hand side (using $\|\alpha\| = \|\beta\| = 1$)
$$\langle \beta^{\otimes k} | \operatorname{tr}_{n \to k} (|\alpha \rangle \langle \alpha |^{\otimes n}) | \beta^{\otimes k} \rangle = \langle \beta^{\otimes k} |\alpha ^{\otimes k} \rangle \langle \alpha ^{\otimes k}| \beta^{\otimes k} \rangle = x^k.$$
For the right hand side of \eqref{eq:inverse-main-equality}, denoting
$$\tilde q(n,k,t):=q(n,k,t) \frac{d[n+k]}{d[n+t]} = (-1)^{t+k}\frac{\binom{n+t}{t}\binom{k}{t}}{\binom{n}{k}},$$
we obtain (see the proof of Theorem \ref{thm:Chiribella} for the combinatorics in the penultimate line):
\begin{align*}
&\sum_{t=0}^k \tilde q(n,k,t) \langle \beta^{\otimes k} | \operatorname{tr}^*_{t \to k}[ \operatorname{MP}_{n \to t} (|\alpha \rangle \langle \alpha|^{\otimes n})] | \beta^{\otimes k} \rangle \\
&\qquad =\sum_{t=0}^k \tilde q(n,k,t) \operatorname{tr}\big[ |\beta \rangle \langle \beta|^{\otimes t}\operatorname{MP}_{n \to t} (|\alpha \rangle \langle \alpha|^{\otimes n})\big]\\
&\qquad =\sum_{t=0}^k \tilde q(n,k,t) d[n+t] \int  |\langle \alpha | \phi \rangle |^{2n} |\langle \beta | \phi \rangle|^{2t}\mathrm{d}\phi\\
&\qquad =\sum_{t=0}^k \tilde q(n,k,t)  \operatorname{tr}\left[P_{sym}^{(n+t,d)} (|\alpha \rangle \langle \alpha|^{\otimes n} \otimes |\beta \rangle \langle \beta|^{\otimes t})\right]\\
&\qquad =\sum_{t=0}^k \tilde q(n,k,t)  \sum_{s=0}^t \frac{\binom{t}{s} \binom{n}{s}}{\binom{n+t}{t}} x^s \\
&\qquad =\sum_{s=0}^k x^s \sum_{t=s}^k \tilde q(n,k,t)   \frac{\binom{t}{s} \binom{n}{s}}{\binom{n+t}{t}}.
\end{align*}
We now compute the inner sum for each $s=0,\ldots,k$ separately. Simple algebraic manipulations and the substitution $t'=k-t$ give
\begin{align*}
\sum_{t=s}^k \tilde q(n,k,t) \frac{\binom{t}{s} \binom{n}{s}}{\binom{n+t}{t}} &=\frac{(k!)^2(n-k)!}{(s!)^2(n-s)!}\sum_{t=s}^k(-1)^{k+t}\frac{1}{(k-t)!(t-s)!}\\
&=\frac{(k!)^2(n-k)!}{(s!)^2(n-s)!}\frac{(-1)^{k+s}}{(k-s)!}\sum_{t'=0}^{k-s}(-1)^{t'}\binom{k-s}{t'}\\
&=\frac{(k!)^2(n-k)!}{(s!)^2(n-s)!}\frac{(-1)^{k+s}}{(k-s)!}\delta_{k,s} = \delta_{k,s}.
\end{align*}
This shows that, for each $\alpha,\beta$, both sides of \eqref{eq:inverse-main-equality} evaluate to the same quantity, namely $x^k$, finishing the proof.
 \end{proof}

We have now all the ingredients to give the proof of our main result. 

\begin{proof}[Proof of Theorem \ref{thm:PSS-C}]
Using Lemma \ref{lem:inverse} and the adjoint of \eqref{equ:MPviaPhi} we have the following equality
\[
\operatorname{tr}^{*}_{k\to n} = \operatorname{tr}^{*}_{k\to n} \circ \Phi^{(n)}_{k\to k}\circ \Psi^{(n)}_{k\to k} = (\operatorname{MP}_{n\to k})^{*} \circ \Psi^{(n)}_{k\to k} =d[n+k]\int   \bra{\phi^{\otimes k}} \Psi^{(n)}_{k\to k}( \cdot)\ket{\phi^{\otimes k}} \ketbra{\phi}{\phi}^{\otimes n} \mathrm{d}\phi,
\]
where the last equality holds since $\mathrm{d}\phi$ is a $(n+k)$-design. Applying $\operatorname{tr}^{*}_{k\to n} \otimes \operatorname{id}_D$ to $W$ and going to the polynomial picture gives thus the following equality (see also Eq.~\eqref{eq:adjoint-of-trace-poly-language}):
\begin{align*}
\|x\|^{2(n-k)}p_W(x,y) &= \|x\|^{2(n-k)} \langle x^{\otimes k} \otimes y | W | x^{\otimes k} \otimes y \rangle \\
&= \langle x^{\otimes n} \otimes y | (\operatorname{tr}^*_{k\to n} \otimes \operatorname{id}_D)(W) | x^{\otimes n} \otimes y \rangle \\
&= d[n+k]\left\langle x^{\otimes n} \otimes y \left| \int   \bra{\phi^{\otimes k}} (\Psi^{(n)}_{k\to k} \otimes \operatorname{id}_D)(W)\ket{\phi^{\otimes k}} \ketbra{\phi}{\phi}^{\otimes n} \mathrm{d}\phi \right| x^{\otimes n} \otimes y \right\rangle \\
&= d[n+k]\int \left\langle \phi^{\otimes k} \otimes y \left|  (\Psi^{(n)}_{k\to k} \otimes \operatorname{id}_D)(W) \right| \phi^{\otimes k} \otimes y \right\rangle |\langle \phi | x \rangle|^{2n}  \mathrm{d}\phi\\
&= \int   p_{\tilde W}(\phi,y) |\braket{\phi}{x}|^{2n}\mathrm{d}\phi,
\end{align*}
where we have set (note the explicit dependence of $\tilde W$, and hence of $p_{\tilde W}$, on the input data $W$)
$$\tilde W := d[n+k] (\Psi^{(n)}_{k\to k} \otimes \operatorname{id}_D)(W)~\in\cB(\vee^k{\mathbb C}^d\otimes{\mathbb C}^D).$$

To conclude, we need to determine when $p_{\tilde W}$ is a positive polynomial. To this end, we insert the expansion of $ \Psi^{(n)}_{k\to k}$ from Lemma \ref{lem:inverse}. This leads to
\begin{align}
p_{\tilde W}(\phi,y) &= d[n+k]\sum_{t=0}^k q(n,k,t)  \bra{\phi^{\otimes k} \otimes y}(\operatorname{tr}^*_{t \to k} \circ \operatorname{tr}_{k \to t} \otimes \operatorname{id}_D)(W)\ket{\phi^{\otimes k} \otimes y} \\
\nonumber &= d[n+k]\sum_{t=0}^k q(n,k,t) \|\phi\|^{2(k-t)} \bra{\phi^{\otimes t} \otimes y}(\operatorname{tr}_{k \to t} \otimes \operatorname{id}_D)(W)\ket{\phi^{\otimes t} \otimes y} \\
\label{eq:proof-PSS-tr-Delta} &= d[n+k]\sum_{t=0}^k q(n,k,t) \|\phi\|^{2(k-t)} p_{(\operatorname{tr}_{k \to t} \otimes \operatorname{id}_D)(W)}(\phi, y)\\
\nonumber & = d[n+k]\sum_{t=0}^k q(n,k,t) \|\phi\|^{2(k-t)}  ((k)_{k-t})^{-2} (\Delta^{k-t}p_W)(\phi, y),\label{last-step-before-applying-bernstein}
\end{align}
where we used Lemma \ref{lem:traceAsDiff} in the last step; note that the (complex) Laplacian acts only on the first set of variables (corresponding to $\phi$). Note also that for $t=k$, the corresponding summand contains $p_W$, and the coefficient $q(n,k,k) = (n+k)!(n-k)!(n!)^{-2}$ is positive. Using Lemma \ref{lem:Bernstein} to upper bound the absolute values of the remaining summands leads to (for $\|\phi\|=\|y\|=1$):
\begin{equation}
p_{\tilde W}(\phi,y) \geq d[n+k]\left(m(W)q(n,k,k) - M(W)\sum_{t=0}^{k-1} |q(n,k,t)| ((k)_{k-t})^{-2}(d/2)^{k-t}(2k)_{2k-2t}\right)
\label{equ:PsiBound}
\end{equation}

In the case $k=1$, this is nonnegative if 
$$m(W)q(n,1,1) \geq d M(W) |q(n,1,0)|,$$
which, after computing $q(n,1,1) = (n+1)/n$ and $q(n,1,0) = -(n+1)/(n(n+d))$, yields \eqref{eq:PSS-bound-n-k1}.

For general $k\geq1$, we bound the negative term in the previous expression from above by the truncation of the Taylor expansion of a certain exponential function (we borrow the idea from the proof of \cite[Theorem 3.11]{reznick1995uniform}). To do this, using the formula for $q(n,k,t)$ from Lemma \ref{lem:inverse}, an elementary computation for $0\leq t\leq k-2$ gives first the following $t$-independent upper bound:
\begin{align}
\nonumber&\frac{|q(n,k,t)|(k-t)!((k)_{k-t})^{-2}(d/2)^{k-t}(2k)_{2k-2t}}{|q(n,k,t+1)|(k-t-1)!((k)_{k-t-1})^{-2}(d/2)^{k-t-1}(2k)_{2k-2t-2}}\\
\label{aux-bound-t}&\qquad\qquad\qquad\qquad\qquad\qquad = \frac{d(t+1)(2t+1) }{n+t+d}\leq \frac{dk(2k-1)}{n+k+d-1}=:r,
\end{align}
where the inequality arises by setting $t=k-1$ in the previous expression, which is an increasing function in $t\in[0,\infty)$; note that the choice $t=k-1$ is sub-optimal, leading to slightly worse but nicer final results, see Remark \ref{rem:better-uglier-bounds}. Applying this estimate repeatedly in the previous sum and changing summation order leads to 
\begin{align*}
\sum_{t=0}^{k-1} |q(n,k,t)|& ((k)_{k-t})^{-2} (d/2)^{k-t}(2k)_{2k-2t}\\
&=\sum_{t=0}^{k-1}\frac{1}{(k-t)!}|q(n,k,t)|  (k-t)!((k)_{k-t})^{-2} (d/2)^{k-t}(2k)_{2k-2t}\\
&\leq \frac{d(2k-1)|q(n,k,k-1)|}{k}\sum_{t=0}^{k-1}\frac{1}{(k-t)!}r^{k-t-1}\\
&=\frac{d(2k-1)|q(n,k,k-1)|}{k}\sum_{s=1}^k\frac{r^{s-1}}{s!}\\
&\leq\frac{d(2k-1)|q(n,k,k-1)|}{k}\frac{e^r-1}{r}.
\end{align*}

Inserting this estimate in \eqref{equ:PsiBound} shows that $p_{\tilde W}(\phi,y)\geq 0$ whenever $n$ satisfies

\begin{align*}
m(W) q(n,k,k) \geq M(W) \frac{d(2k-1)|q(n,k,k-1)|}{k}\frac{e^r-1}{r}\\
\iff \frac{k^2d(2k-1)}{(n+k+d-1)k}\frac{M(W)}{m(W)} \frac{e^r-1}{r} \leq 1.
\end{align*}
Using the fact that the function $r \mapsto (e^r-1)/r$ is increasing, we find that it is sufficient that $n$ satisfies, for some $\Gamma >0$ to be determined later,
\begin{align*}
\frac{d k(2k-1)}{n+k+d-1} \frac{M(W)}{m(W)} \frac{e^\Gamma-1}{\Gamma} &\leq 1\\
\text{and}\qquad r= \frac{dk(2k-1)}{n+k+d-1} &\leq \Gamma
\end{align*}
in order for $p_{\tilde W}$ to be non-negative. Re-arranging terms, we obtain the following two sufficient conditions:
\begin{align*}
n &\geq dk(2k-1) \frac{e^\Gamma-1}{\Gamma} \frac{M(W)}{m(W)} - d-k+1 \\
n &\geq dk(2k-1) \frac{1}{\Gamma} - d-k+1.
\end{align*}
We now choose $\Gamma$ such that the two inequalities are identical,
$$\Gamma = \ln \left( 1 + \frac{m(W)}{M(W)} \right),$$
and the bound (\ref{eq:PSS-bound-n-other}) follows. That the bound (\ref{eq:PSS-bound-n-k1}) is better in the case $k=1$, can be easily seen from the inequality $\ln(1+m/M)\leq m/M$.
\end{proof}

\begin{remark}\label{rem:better-uglier-bounds}
When $k\geq2$, any $n\geq k$ with 
\begin{equation}\label{eq:PSS-bound-n}
n \geq \frac{d(k-1)(2k-3)}{\ln \left( 1 + \frac{(k-1)(2k-3)}{k(2k-1)}\frac{m(W)}{M(W)} \right)} - d-k+2
\end{equation}
(instead of Eq.\ (\ref{eq:PSS-bound-n-other})) is sufficient for the conclusions of Theorem \ref{thm:PSS-C}. In other words, Eq.\ \eqref{eq:PSS-bound-n} is a better bound (a weaker requirement on $n$) than \eqref{eq:PSS-bound-n-other}. That \eqref{eq:PSS-bound-n} suffices, follows as in the proof of Theorem \ref{thm:PSS-C} when using $d(k-1)(2k-3)/(n+d+k-2)$ as the upper bound in (\ref{aux-bound-t}), which is obtained by setting $t=k-2$ in the preceding expression, which is the best bound one can obtain using the monotonicity in $t$. The proof proceeds by replacing $r:=d(k-1)(2k-3)/(n+d+k-2)$, which still has to satisfy $r\leq\Gamma$. Setting $\Gamma = \ln \left( 1 + \frac{(k-1)(2k-3)}{k(2k-1)}\frac{m(W)}{M(W)} \right)$, one sees that the bound for $n$ from \eqref{eq:PSS-bound-n} suffices.
\end{remark}

\begin{remark}[Less stringent bounds on $n$]\label{bounds-for-k-geq-2}
Our bounds (\ref{eq:PSS-bound-n}) and (\ref{eq:PSS-bound-n-k1}) on $n$ are better (less stringent) than the ones from \cite[Theorem 1]{to2006effective}, even by roughly a factor of $2$ for the case $k=1$ in Eq.\ (\ref{eq:PSS-bound-n-k1}). This is because we used a better Bernstein-type inequality in the complex case (our Lemma \ref{lem:Bernstein}) than these authors.

Even less stringent bounds on $n$ may be found, for given $m(W)$, $M(W)$, $d$ and $k$, in a numerical way, namely by searching for the smallest $n\geq k$ such that the expression in square brackets in (\ref{equ:PsiBound}) becomes nonnegative; eqs.\ (\ref{eq:PSS-bound-n-k1}) and (\ref{eq:PSS-bound-n}) give a guarantee for when the search has to terminate. Note that the everything in the proof after Eq.\ (\ref{equ:PsiBound}) was devoted merely to derive the simple analytical expressions given in (\ref{eq:PSS-bound-n-k1}) and (\ref{eq:PSS-bound-n}) as sufficient bounds on $n$.

One can even obtain somewhat better lower bounds on $p_{\tilde W}(\phi,y)$ than the one given in (\ref{equ:PsiBound}) using an idea of \cite{nam2008effective}, and these can again be used in analytical \cite[Section 3.2]{nam2008effective} or numerical ways to obtain sufficient bounds on $n$ for a positive representation. We compare all these analytical and numerical bounds for a real version of this result in Example \ref{exa:Motzkin-R}.
\end{remark}

\begin{remark}\label{rem:sos-form}
Theorem \ref{thm:PSS-C} guarantees $p_{\tilde W}(\phi,y)$ to be nonnegative for any $\phi,y$ if only $n$ is sufficiently big. As $p_W$ is of homogeneous degree $1$ in $y$ and $\overline{y}$, we can write $p_{\tilde W}(\phi,y)=\langle y|\tilde{W}_\phi|y\rangle$ where each matrix $\tilde{W}_\phi\in\cH(\mathbb C^d)$ is positive semidefinite. With the eigendecompositions ${\tilde W}_\phi=\sum_{i=1}^D|\tilde w_\phi^{(i)}\rangle\langle\tilde w_\phi^{(i)}|$ this becomes $p_{\tilde W}(\phi,y)=\sum_{i=1}^D|\langle \tilde w_\phi^{(i)}|y\rangle|^2$ for any $\phi$, $y$. Inserting this into (\ref{eq:sos-integral-repn}) and using the construction of discrete, finitely supported complex spherical designs (Appendix \ref{spherical-design-appendix}), this shows constructively the existence of a sums-of-squares decomposition of the following special form (cf.\ \eqref{eq:sos-integral-repn}):
\begin{equation}
\|x\|^{2(n-k)}p_W(x,y) = \sum_\phi\sum_{i=1}^D|\braket{\phi}{x}|^{2n}|\braket{w_\phi^{(i)}}{y}|^2,
\end{equation}
where the sum over the unit vectors $\phi$ of the finite spherical design contains at most $(n+k+1)^{2d}$ terms and we have taken the weights into the normalization of the vectors $w_\phi^{(i)}$. The obtained sum-of-squares decomposition is thus very special as each term is a $2n$-th power of the absolute value of a linear form in $x$ multiplied with the absolute square of a linear form in $y$.
\end{remark}

\section{The case of real polynomials}\label{sec:PSS-R}

We turn now to the case of real polynomials, the classical setting of the study of positive polynomials and the corresponding Positivstellens\"atze \cite{quillen1968representation,schmuedgen1991positivstellensatz,putinar1993positive,reznick1995uniform}. It turns out that the proof strategy we developed for the complex case can be adapted to the real situation, yielding a small improvement of	  Reznick's Positivstellensatz \cite[Theorem 3.12]{reznick1995uniform}. We shall outline the main steps below, but we now point out two important facts. The first one is that, in the real case, the objects and maps we need to set up do not have a direct interpretation in the language of quantum information (which is a theory built on the field of complex numbers). The second point is that our derivation in the real case follows closely the proof strategy from \cite{reznick1995uniform}, once one moves from the language of polynomials to that of linear algebra; we shall emphasize which are the similarities and the (small) differences as we move on. 

Let us first explain in detail the relation between the symmetric subspace and homogeneous polynomials in the real case. We denote by $\mathcal{H}_{n}(\R^d)$ the vector space of real homogeneous polynomials in $d$ real variables, of degree $n$. To any symmetric vector $v\in \vee^{n}\R^d$ we associate a homogeneous polynomial $p_v\in \mathcal{H}_{n}(\R^d)$ of degree $n$ given by
\[
p_v(x) = \braket{x^{\otimes n}}{v},
\]
for $x\in \R^d$. Note that the correspondence between $\mathcal{H}_{n}(\R^d)$ and $\vee^{n}\R^d$ introduced here is one-to-one, and we shall often switch between the ``polynomial'' and the ``linear algebra'' viewpoints. As an important example, consider the case of the square of the norm, which corresponds to the (un-normalized) \emph{maximally entangled state} $\Omega_d$: 
$$\|x\|^{2} = \sum_{i=1}^d x_i^2 \in \mathcal{H}_2(\R^d)  \qquad \longleftrightarrow \qquad \ket{\Omega_d} = \sum_{i=1}^d e_i \otimes e_i \in \vee^2 \R^d.$$
Since we are interested in positive polynomials, we shall focus in what follows on the case of polynomials of even degree. By the correspondence above, a linear map $F:\vee^{2n}\R^d \to \vee^{2k}\R^d$ corresponds to a linear map $\tilde F:\mathcal{H}_{2n}(\R^d)\to \mathcal{H}_{2k}(\R^d)$ and in the following we shall abuse notation by writing $F(v)$ and $F(p)$ for $p=p_v$ interchangeably. On the space $\vee^{2n}\R^d$ we can define a ``partial trace'' operation $\operatorname{tr}_{n\to k}:\vee^{2n}\R^{d}\to \vee^{2k}\R^{d}$ (we assume here $k \leq n$) via 
\[
\operatorname{tr}_{n\to k}(v) = \langle\Omega_d^{\otimes (n-k)}, v\rangle = (\Omega_d^{\otimes (n-k)} \otimes I_d^{\otimes 2k})^*v,
\]
where $I_d\in \M_d(\R)$ is the identity matrix. On $\mathcal{H}_{2n}(\R^d)$ we can define the real Laplacian $\Delta_\R = \sum^d_{i=1}\frac{\partial^2}{\partial x^2_i}$. When applied to polynomials, it turns out that the map $\operatorname{tr}_{n\to k}$ is, up to a constant, an iterated Laplacian:

\begin{lemma}[Partial trace vs. Laplacian]
For any $p\in \mathcal{H}_{2n}(\R^d)$ we have
\[
(2n)_{2(n-k)}\operatorname{tr}_{n\to k}(p) = \Delta_\R^{n-k}p.
\]
\label{lem:TraceLapReal}
\end{lemma}
\begin{proof}
Recall that the set $\{ {v}^{\otimes 2n}: {v}\in \R^d\}$ spans $\vee^{2n}\R^d$ \cite[Theorem 3]{harrow2013church}. Therefore, the corresponding set of polynomials $\{ p_{{v}^{\otimes 2n}}(x) = \braket{x}{v}^{2n}\}_{v \in \R^d}$ spans $\mathcal{H}_{2n}(\R^d)$ and it will be sufficient to show the lemma for this spanning set. Note that on one hand
\[
\Delta_\R p_{{v}^{\otimes 2n}}(x) = \Delta_\R[(v_1x_1 + \cdots v_dx_d)^{2n}] = 2n(2n-1)\|v\|^2\braket{x}{v}^{2n-2} = 2n(2n-1)\|v\|^2 p_{{v}^{\otimes (2n-2)}}(x).
\]
On the other hand 
\[
\operatorname{tr}_{n\to (n-1)}({v}^{\otimes 2n}) = \|v\|^2{v}^{\otimes (2n-2)}.
\]
Direct comparison of the two expressions shows that
\[
(2n)_{2}\operatorname{tr}_{n\to (n-1)}(p) = \Delta_\R p .
\] 
Finally, by iterating the previous formula the statement of the lemma follows.
\end{proof}

Using the Hilbert space structure of $\vee^{2n}\R^d$ we can introduce the dual map of $\operatorname{tr}_{n\to k}$, $\operatorname{tr}^*_{k\to n}:\vee^{2k}\R^d\to \vee^{2n}\R^d$. It is easy to express this map in the ``polynomial picture''. For this consider $v\in \vee^{2k}\R^d$ and compute 
\[
\operatorname{tr}^*_{k\to n}(p_{v}(x)) = \braket{x^{\otimes 2n}}{\operatorname{tr}^*_{k\to n}(v)} = \braket{\operatorname{tr}_{n\to k}(x^{\otimes 2n})}{v} = \|x\|^{2(n-k)}p_v(x).
\]
In other words, the map $\operatorname{tr}^*$ (which is related to the cloning operation from quantum information theory in the complex setting) corresponds to multiplying a polynomial with an even power of the euclidean norm of the variable vector. We present the correspondence between symmetric tensors and homogeneous real polynomials, as well as the dual operations of multiplying with the norm and taking the Laplacian in Figure \ref{fig:diagrams-polynomials-real}.

\begin{figure}
	\centering
	\begin{minipage}{6in}
		\centering
		\raisebox{-0.5\height}{\includegraphics[width=0.2\linewidth]{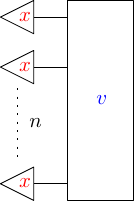}}
		\hspace*{.9in}
		\raisebox{-0.5\height}{\includegraphics[width=0.2\linewidth]{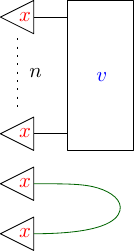}}
		\hspace*{.9in}
		\raisebox{-0.5\height}{\includegraphics[width=0.2\linewidth]{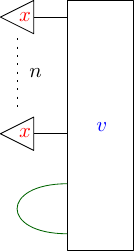}}	
	\end{minipage}	
	\caption{Graphical representation of the correspondence between symmetric tensors $\textcolor{blue}v$ and homogeneous polynomials. From left to right, we have depicted the diagrams for $p_{\textcolor{blue}{v}}(\textcolor{red}{x})$, $\textcolor{darkgreen}{\|\textcolor{red}{x}\|^{2}}p_{\textcolor{blue}{v}}(\textcolor{red}{x})$, and $\textcolor{darkgreen}{(n+2)_2 \Delta} p_{\textcolor{blue}{v}}(\textcolor{red}{x})$ respectively, where $\Delta$ is the real Laplacian. Note that multiplying with the norm and the Laplacian are, up to constants, dual operations.}
	\label{fig:diagrams-polynomials-real}
\end{figure}

Let us now introduce the measure and prepare map in the real case: 

\begin{align*}
\mathrm{MP}_{n \to k}^\R : \vee^{2n}\R^d &\to \vee^{2k}\R^d\\
v &\mapsto d_\R[n+k] \int_{\mathcal S_\R^{d-1}} \langle v, \phi^{\otimes 2n}\rangle \ket{\phi}^{\otimes 2k} \mathrm{d}\phi,
\end{align*}
with the normalization constant
\[
d_\R[n] = \frac{2^{2n}n!~\Gamma(n+d/2)}{(2n)!~\Gamma(d/2)}.
\]

The choice of the normalization function is motivated by the following Chiribella-like identity; the proof is similar to the complex case, see Appendix \ref{sec:app-HI-R} for the corresponding Hilbert identity.
$$
\operatorname{MP}_{n\to k}^\R = \sum^{\min(n,k)}_{s=0} c_\R(n,k,s) \operatorname{tr}^{*}_{s\to k}\circ \operatorname{tr}_{n\to s}$$
where 
$$c_\R(n,k,s) = 2^{2s} \frac{\binom{n+k}{2s,n-s,k-s}}{\binom{2n+2k}{2k}}.$$

As in the complex case, we define the map
\[\Phi^{\R,(n)}_{k \to k} :=   \sum_{s=0}^k c_\R(n,k,s) \operatorname{tr}^*_{s \to k} \circ \operatorname{tr}_{k \to s}.\]
This map is related to the measure-and-prepare map by the relation
$$\operatorname{MP}_{n \to k}^\R = \Phi^{\R,(n)}_{k \to k} \circ \operatorname{tr}_{n \to k}.$$
This is the decomposition \cite[Theorem 3.7]{reznick1995uniform} in Reznick's work. It is based on Hobson's identity \cite{hobson1931theory}, which can be seen as the real, polynomial analogue of the Chiribella identity (cf.~Theorem \ref{thm:Chiribella}). Its compositional inverse is given by
$$\Psi^{\R,(n)}_{k \to k} :=   \sum_{t=0}^k q_\R(n,k,t)  \operatorname{tr}^*_{t \to k} \circ \operatorname{tr}_{k \to t},$$
with coefficients 
$$q_\R(n,k,t) = (-1)^{k+t} 2^{-2k} \frac{\binom{2n+2t}{2t} \binom{2k}{k-t} }{\binom{n+t}{k+t}} \frac{d_\R[n+t]}{d_\R[n+k]},$$
such that 
\[
\sum^k_{t=s} q_\R(n,k,t) c_\R(n,t,s)\frac{d_\R\lbr n+k\rbr}{d_\R\lbr n+t\rbr} = \delta_{k,s}.
\]
In Reznick's derivation, this corresponds to \cite[Theorem 3.9]{reznick1995uniform}. We claim that our linear algebraic language is more elegant, but ultimately the two formulations are equivalent. 

The proof of Theorem \ref{thm:PSS-C} can be easily adapted to the real case, as follows, yielding a new variant of Reznick's real Positivstellensatz  \cite[Theorem 3.12]{reznick1995uniform}. To estimate the extreme values of derivatives of polynomials we shall need (real) Bernstein inequality, see Lemma \ref{lem:Bernstein-real}; we would like to stress again that this inequality is the only analytical tool used in the proof, the rest being basic linear algebra.

\begin{theorem}[Positivstellensatz, real case]\label{thm:PSS-R}
Consider $v\in \vee^{2k} \R^d \otimes \R^D$ with $m(v) >0$. Then for any $n \geq k$ such that
\begin{equation}\label{eq:PSS-bound-n-R}
2n \geq \frac{dk(2k-1)}{\ln \left( 1 + \frac{m(v)}{M(v)} \right)} +2 -2k-d,
\end{equation}
we have
\begin{equation}\label{eq:sum-of-2n-powers-R}
\|x\|^{2(n-k)}p_v(x,y) = \int d\phi  p_{\tilde v}(\phi,y) \braket{\phi}{x}^{2n}
\end{equation}
with $p_{\tilde v}(\phi,y) \geq 0$ for all $\phi \in \R^d$ and $y \in \R^D$. Therefore, $\|x\|^{2(n-k)}p_v(x,y)$ is a sum of squares. In the case $k=1$, the bound \eqref{eq:PSS-bound-n-R} can be improved to 
\begin{equation}\label{eq:PSS-bound-n-k1-R}
2n \geq d \frac{M(v)}{m(v)} -d.
\end{equation}
\end{theorem}
\begin{proof}
The proof idea is identical to that of Theorem \ref{thm:PSS-C}, however some details and coefficient values are the different. The first difference appears in equation \eqref{eq:proof-PSS-tr-Delta}, where the polynomial associated to the partial trace of $v$ is related to the Laplacian of $v$ by the following formula, proven in Lemma \ref{lem:TraceLapReal}
$$p_{(\operatorname{tr}_{k \to t} \otimes \operatorname{id}_D)(v)}(x,y) = ((2k)_{2(k-t)})^{-1} \Delta_\R^{k-t}p_v(x,y).$$
In the equation above, $p_v$ is a homogeneous polynomial of total degree $(2k,1)$ in the real variables $x_1, \ldots, x_d$ and $y_1, \ldots, y_D$, and $\Delta_\R$ is the ``real'' Laplacian, $\Delta_\R = \sum_{i=1}^d \partial^2 / \partial x_i^2$ acting on the first set of variables $x_i$. Using the usual Bernstein inequality (Lemma \ref{lem:Bernstein-real}), we get
$$|((2k)_{2(k-t)})^{-1} \Delta_\R^{k-t}p_v| \leq d^{k-t} M(v).$$
Plugging this into the expression for $p_{\tilde v}$, we obtain the following relation (this corresponds to \eqref{equ:PsiBound} in the complex case)
\begin{equation}
p_{\tilde v}(\phi,y) \geq d_\R[n+k] \left( m(v)q_\R(n,k,k) - M(v)\sum_{t=0}^{k-1} |q_\R(n,k,t)| d^{k-t}\right)
\label{eq:PsiBound-R}
\end{equation}
Exactly as in the complex case we bound the quotients of consecutive $q_\R$'s by 
\begin{equation}
\frac{|q_\R(n,k,t)|(k-t)!}{|q_\R(n,k,t+1)|(k-t-1)!} = \frac{(t+1)(2t+1)}{(2n+2t+d)}\leq \frac{k(2k-1)}{2n+2k+d-2},
\label{eq:ConsQBoundReal}
\end{equation}
where the last inequality comes from setting $t=k-1$ in the preceeding expression (using monotonicity). Again, we employ an exponential function approximation to derive the final estimate. Here, our proof strategy diverges qualitatively from the one in \cite[Theorem 3.12]{reznick1995uniform}, yielding the small improvement discussed in Remark \ref{rem:PSS-R-vs-Reznick}. We leave the computational details to the reader.
\end{proof}
\begin{remark}
Similarly as for Theorem \ref{thm:PSS-C} (cf.\ Remark \ref{bounds-for-k-geq-2}), in the case $k \geq 2$ one can obtain the following improved lower bound for $n$:
\begin{equation}\label{eq:PSS-bound-n-R-modified}
2n \geq \frac{d(k-1)(2k-3)}{\ln \left( 1 + \ \frac{(k-1)(2k-3)}{k(2k-1)}\frac{m(v)}{M(v)} \right)} +4-2k-d.
\end{equation}
We prefer however the bound in the statement, which has a more compact form. 
\end{remark}
\begin{remark}\label{rem:PSS-R-vs-Reznick}
The result above is to be compared with Renzick's Positivstellensatz \cite[Theorem 3.12]{reznick1995uniform}, which, in our notation, gives the bound
\begin{equation}\label{eq:PSS-bound-n-R-Reznick}
2n \geq \frac{dk(2k-1)}{\ln 2} \frac{M(v)}{m(v)} -d.
\end{equation}
Although the two bounds have the same leading orders in $d$ and $k$, our lower bound is smaller, since the constant in front of the leading term $dk^2$ is smaller: in the regime where $m(v)/M(v) \ll 1$, the ratio between the two lower bounds is $\ln 2<1$.
\end{remark}

We would now like to discuss the bounds above in a concrete situation in order to get an idea about the optimality of the lower bounds. 

\begin{example}\label{exa:Motzkin-R}
Consider the celebrated Motzkin polynomial to which we add a positive multiple of the norm:
\begin{equation}\label{eq:shifted-Motzkin}
p_\varepsilon(x,y,z) = x^4y^2+y^4z^2+z^4x^2-3x^2y^2z^2 + \varepsilon(x^2+y^2+z^2).
\end{equation}
Obviously, we have $m(p_\varepsilon) = \varepsilon$; using Lagrange multipliers, one can easily find  $M(p_\varepsilon) = \varepsilon + 4/27$. In Figure \ref{fig:bounds-1-2-q}, we compare three lower bounds on $n$. In the left panel, we plot the bounds from equations \eqref{eq:PSS-bound-n-R} and \eqref{eq:PSS-bound-n-R-modified}, observing that the latter performs better (here, $d=k=3$). In the right panel, we compare the bound from \eqref{eq:PSS-bound-n-R-modified} with the one obtained by working out the value of $n$ directly from equation \eqref{eq:PsiBound-R}. Not that this latter bound is necessarily better, since is shortcuts the second part of the proof of Theorem \ref{thm:PSS-R}.

\begin{figure}[htbp]
\begin{center}
\includegraphics[scale=0.55]{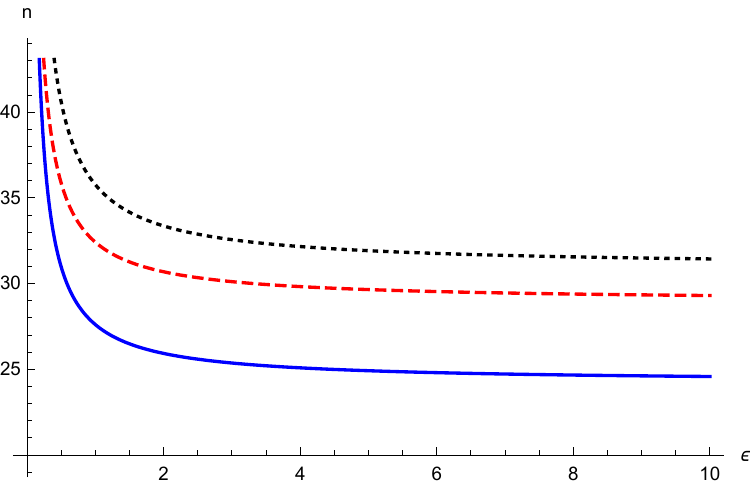} \quad  \includegraphics[scale=0.55]{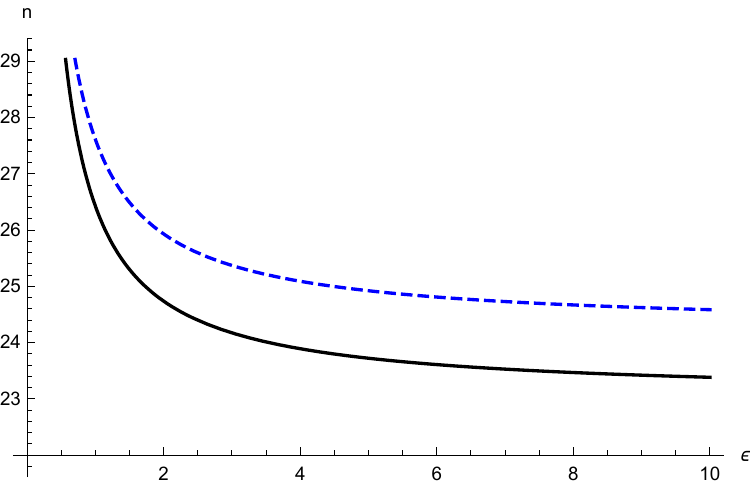}
\caption{Comparing lower bounds for $n$ needed for the real Positivstellensatz to hold for the shifted Motzkin polynomial from \eqref{eq:shifted-Motzkin}, as a function of $\varepsilon$. In the left panel, Reznick's bound \eqref{eq:PSS-bound-n-R-Reznick} (black, dotted curve), the bounds from \eqref{eq:PSS-bound-n-R} (red, dashed curve) and \eqref{eq:PSS-bound-n-R-modified} (blue curve) are plotted. On the right, the bound obtained from equation \eqref{eq:PsiBound-R} (black curve) is plotted against the bound from \eqref{eq:PSS-bound-n-R-modified} (blue, dashed curve).}
\label{fig:bounds-1-2-q}
\end{center}
\end{figure}

In Figure \ref{fig:boundq-testW}, we compare again the bound obtained by working out the value of $n$ directly from equation \eqref{eq:PsiBound-R} with the one obtained by asking that the coefficients of the polynomial 

\begin{equation}\label{eq:def-p-n-eps}
p_{n,\varepsilon}(x,y,z):=(x^2+y^2+z^2)^{n-3}p_\varepsilon(x,y,z)
\end{equation}
be non-negative. Indeed, if $n,\varepsilon$ are such that \eqref{eq:sum-of-2n-powers-R} holds, then the $[2p,2q,2r]$ coefficient of $p_{n,\varepsilon}$ reads (here, $p+q+r=n$)
$$ \int d\phi  \tilde p(\phi) \binom{2n}{2p,2q,2r} \phi_x^{2p}\phi_y^{2q}\phi_z^{2r} \geq 0.$$ 
Hence, for each fixed $n$, we can find numerically the smallest constant $\varepsilon_n >0$ such that, for all $\varepsilon \geq \varepsilon_n$, all coefficients of $p_{n,\varepsilon}$ are non-negative. Note that all monomials of $p_{n,\varepsilon}$ are of the form $c_{2p, 2q, 2r} x^{2p} y^{2q} z^{2r}$ for some non-negative $p,q,r$ with $p+q+r = n$, see \eqref{eq:def-p-n-eps}. 

\begin{figure}[htbp]
\begin{center}
\includegraphics[scale=0.85]{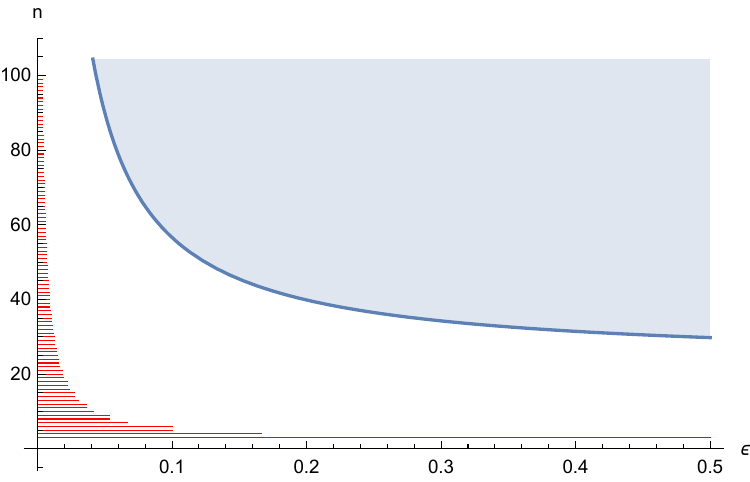}
\caption{Lower and upper bounds for the minimal $n$ such that the decomposition \eqref{eq:sum-of-2n-powers-R} holds for the shifted Motzkin polynomial \eqref{eq:shifted-Motzkin}, as functions of $\varepsilon$. The upper bound (filled blue region) comes from equation \eqref{eq:PsiBound-R}. The lower bounds (red horizontal bars, one for each value of $n$) come from requiring that all the coefficients of the polynomial $p_{n,\varepsilon}$ are non-negative.}
\label{fig:boundq-testW}
\end{center}
\end{figure}

\end{example}

\section{Application to exponential de Finetti theorems}\label{sec:exp-deF}

We show in this section how the so-called \emph{exponential de Finetti theorem} \cite{renner2007symmetry,koenig2009most} follows from the analysis of the inversion of the Chiribella identity from Lemma \ref{lem:inverse}. A similar derivation can be found in \cite[Theorem 8]{harrow2013church}; our result improves on this by having explicit constants in front of the maps, and thus achieving better error terms. The main idea here is that we want to approximate marginals of symmetric states not by states which are exactly tensor powers of pure states (as in the usual de Finetti theorem), but with states from the larger set
$$\mathcal W_{r} := \bigcup_{|\phi \rangle \in \mathbb C^d} \operatorname{span}\{P_{sym}^{(k)} |\phi\rangle^{\otimes k-r} \otimes |\psi \rangle \, : \, |\psi \rangle \in (\mathbb C^d)^{\otimes r}\}.$$
Such states are called \emph{$(k,r,d)$-almost product states}, and they form a larger set than the class of product states. Considering such a larger set of targets states allows for faster convergence in de Finetti-type result: one can go from linear to exponential convergence speed using this relaxation. De Finetti type theorems have found many applications in (quantum) information theory, mainly to reduce the analysis of protocols where symmetry plays an important role to that of the much simpler i.i.d.~protocols \cite{renner2007symmetry}. 

The main technical insight here is that almost product states obviously lie in the ranges of the maps $\operatorname{tr}^*_{k-s \to k} \circ \operatorname{MP}_{n \to k-s}$,
for all $0 \leq s \leq r$. This leads to the idea that one has to truncate the sum expression of the partial trace operator not only to the first term (which is the case for the Positivstellensatz in Theorem \ref{thm:PSS-C}), but to the $r$-th term.

\begin{theorem}
Let $1 \leq k <n$ be positive integers such that
$$\delta:=\frac{k(k+d-1)}{n+k+d-1}<\frac 1 3 \iff k < \frac{\sqrt{12n+(3d-2)^2}+4-3d}{6},$$
and let $\operatorname{\widetilde{MP}}$ denote the measure-and-prepare map, rescaled to be a quantum channel
$$\operatorname{\widetilde{MP}}_{n \to k}(X) = d[n] \int_{\mathcal S^{d-1}} \langle \varphi^{\otimes n} | X | \varphi^{\otimes n} \rangle \ketbra{\varphi}{\varphi}^{\otimes k} \mathrm{d}\varphi.$$
Then, for any $0 \leq r \leq k$, we have the following estimate in diamond norm
$$\left\| \operatorname{tr}_{n \to k} - \sum_{s=0}^r \hat q(n,k,k-s) \operatorname{Clone}_{k-s \to k} \circ \operatorname{\widetilde{MP}}_{n \to k-s} \right\|_\diamond \leq \varepsilon_r,$$
where the error is bounded by 
$$\varepsilon_r \leq \frac{\delta^{r+1}}{1-3\delta},$$
and where 
$$ \hat q(n,k,k-s) = q(n,k,k-s)\frac{d[k]}{d[k-s]},$$
with $q(n,k,t)$ as in Lemma \ref{lem:inverse}. In particular, the $k$-body marginal of a $n$-symmetric state is $\varepsilon_r$ away (in $1$-norm) from a linear combination of projections on $(k,r,d)$-almost product states.
\end{theorem}
\begin{proof}
The second claim follows from the first one and fact the range of the quantum channels $\operatorname{Clone}_{k-s \to k} \circ \operatorname{\widetilde{MP}}_{n \to k-s}$ is precisely the set of density matrices supported on the span of $\mathcal W_r$, for all $0 \leq s \leq r$. Indeed, this is a simple consequence of the form of the channels $\operatorname{Clone}_{k-s \to k}$ and $\operatorname{\widetilde{MP}}_{n \to k-s}$.

Regarding the main inequality, the starting point is the inverse of the Chiribella identity \eqref{eq:inverse-main-equality} proven in Lemma \ref{lem:inverse}:
\begin{align*}
\operatorname{tr}_{n \to k} &=  \sum_{s=0}^k q(n,k,k-s)\frac{d[n+k]d[k]}{d[n]d[k-s]} \operatorname{Clone}_{k-s \to k} \circ \operatorname{\widetilde{MP}}_{n \to k-s}\\
&= \sum_{s=0}^k \hat q(n,k,k-s) \operatorname{Clone}_{k-s \to k} \circ \operatorname{\widetilde{MP}}_{n \to k-s}.
\end{align*}
The inequality in the statement is obtained by bounding the diamond norm of the tail of the sum above using the triangle inequality and the fact that both $\operatorname{Clone}_{k-s \to k}$ and $\operatorname{\widetilde{MP}}_{n \to k-s}$ are quantum channels:
$$\varepsilon_r = \sum_{s=r+1}^k |\hat q(n,k,k-s)|.$$
The claimed bound on $\varepsilon_r$ follows from the geometric sum formula and the bound on the coefficients, which we show next. Compute first, using \eqref{eq:def-q},
$$\frac{|\hat q(n,k,k-s-1)|}{|\hat q(n,k,k-s)|} = \frac{k+d-1-s}{n+k+d-1-s} \cdot \frac{k-s}{s+1} \leq \frac{k(k+d-1)}{n+k+d-1} =: \delta,$$
where we have used the fact that both functions 
$$s \mapsto \frac{k+d-1-s}{n+k+d-1-s} \quad \text{ and } \quad s \mapsto \frac{k-s}{s+1}$$
are decreasing on $[0,r]$. The $s=0$ term is bounded as follows:
\begin{align*}
\frac{1}{\hat q(n,k,k)} &= \frac{(n)_k}{(n+k+d-1)_k} = \frac{n \cdots (n-k+1)}{(n+k+d-1) \cdots (n+d)}\\
&\geq \left( \frac{n-k+1}{n+d} \right)^k = \left( 1 - \frac{k+d-1}{n+d} \right)^k\\
&\geq 1 - \frac{k(k+d-1)}{n+d} = 1- \delta \left(1+\frac{k-1}{n+d} \right),
\end{align*}
while the general term satisfies
$$|\hat q(n,k,k-s)| \leq \frac{\delta^s}{1-\delta\left(1 + \frac{k-1}{n+d}\right)} \leq \frac{\delta^s}{1-2\delta}.$$
The total error is bounded by (we use $\delta < 1/3$)
$$\varepsilon_r \leq \sum_{s=r+1}^k \frac{\delta^s}{1-2\delta} < \sum_{s=r+1}^\infty \frac{\delta^s}{1-2\delta} =  \frac{\delta^{r+1}}{(1-\delta)(1-2\delta)} \leq \frac{\delta^{r+1}}{1-3\delta},$$
concluding the proof. 
\end{proof}

\begin{remark}
	Using the same type of exponential sum estimates as in the proof of Theorem \ref{thm:PSS-C}, one obtains the following bound on the error $\epsilon_r$, unconditionally on the values of the parameters $d$ and $1 \leq k \leq n$: 
	$$\epsilon_r \leq \frac{e^\delta\left(1 -\frac{\Gamma(r+1,\delta)}{r!}\right)}{\left(1 - \frac{k+d-1}{n+d}\right)^k},$$
	where $\Gamma$ is the upper incomplete Gamma function
	$$\Gamma(s,x) = \int_{x}^\infty t^{s-1}e^{-t} \mathrm{d}t.$$
\end{remark}

\begin{remark}
A similar result holds in the real case, where one has to replace the coefficients $q$ and $d[\cdot]$ by $q_\R$ and $d_\R[\cdot]$. The proof steps are identical, and one finds the following value for the base of the exponential
$$\delta_\R = \frac{k(2k+d-2)}{2n+2k+d-2}.$$
We leave the details of the calculation to the reader. 
\end{remark}

\bigskip
\textbf{Acknowledgments.} The authors would like to thank the Mathematisches Forschungsinstitut Oberwolfach (MFO) for providing a perfect working environment: this work was almost entirely done during a \emph{Research in Pairs} stay at the MFO. A.M-H acknowledges financial support from VILLUM FONDEN via the QMATH Centre of Excellence (Grant No. 10059) and from the European Union’s Horizon 2020 research and innovation programme under the Marie Sk\l odowska-Curie Action TIPTOP (grant no. 843414). I.N.'s research has been supported by the ANR projects {StoQ} {ANR-14-CE25-0003-01} and {NEXT} {ANR-10-LABX-0037-NEXT}. We thank Hamza Fawzi for pointing us to relevant literature. The authors would also like to thank Birgit Kraft for inspiring discussions.

\appendix

\section{Gaussian integrals and Hilbert identities}\label{sec:app-GI-HI}

In this appendix, we discuss Hilbert identities, that is expressions of powers of the real or complex Euclidean norms in a $d$-dimensional space as sums of powers of linear forms. These identities date back to Hilbert's work \cite{hilbert1909beweis} on Waring's problem \cite{ellison1971waring}. Our approach here is probabilistic, and largely inspired by \cite{harrow2013church}. We present separately the real and the complex cases. 

\subsection{The real case}\label{sec:app-HI-R}

\begin{definition}[Pairings]
A permutation $\pi \in S_{2n}$ is called a pairing iff there are transpositions with \emph{disjoint supports} $\tau_i$ for each $i\in \{ 1,\ldots ,n\}$ with $\pi = \tau_1\circ \cdots \circ \tau_n$. We denote by $\Pi[2n]$ the set of all pairings.
\end{definition}

\begin{lemma}[Combinatorial formula for Gaussian integrals -- Real case]
Let $\ket{\hat{\phi}}:\Omega\to \R^d$ be a vector-valued random variable where the entries $\hat{\phi}_i$ are i.i.d.~Gaussian distributed with mean $0$ and variance $1$. Then we have 
\[
\mathbb{E}[\ket{\hat{\phi}}^{\otimes 2n}] = \sum_{\pi\in \Pi[2n]} \ket{\pi},
\] 
where $\ket{\pi} = \bigotimes^n_{i=1} \ket{\Omega_{\tau_i}}$ for the transpositions $\tau_i$ with disjoint supports such that $\pi = \tau_1\circ \cdots \circ \tau_n$.
\label{lem:GaussianCombRe}
\end{lemma}

\begin{proof}
We find that in the computational basis
\[
\bra{i_1i_2\ldots i_{2n}}\mathbb{E}[\ket{\hat{\phi}}^{\otimes 2n}] = \mathbb{E}[\hat{\phi}_{i_1}\ldots \hat{\phi}_{i_{2n}}] = \sum_{\pi\in \Pi[2n]}\prod_{(r,s)\in \pi} \mathbb{E}[\hat{\phi}_{i_r}\hat{\phi}_{i_s}] 
\]
where we used Wick's (or Isserlis') formula \cite{isserlis1918formula} in the last step and we write $(r,s)\in \pi$ for the transpositions $\tau_i = (r,s)$ such that $\pi = \tau_1\circ \cdots \circ \tau_n$. By our assumptions on the $\hat{\phi}$ the above expression simplifies to
\[
\bra{i_1i_2\ldots i_{2n}}\mathbb{E}[\ket{\hat{\phi}}^{\otimes 2n}] = \sum_{\pi\in \Pi[2n]}\prod_{(r,s)\in \pi} \delta_{i_r i_s} = \bra{i_1i_2\ldots i_{2n}}\sum_{\pi\in \Pi[2n]}\ket{\pi}.
\]

\end{proof}

In the following let $\mathcal S^{d-1}_\R\subset \R^d$ denote the real unit sphere. 

\begin{lemma}[Real Spherical Hilbert identity -- Linear algebra form]
We have
\[
\frac{1}{|\Pi[2n]|}\sum_{\pi\in \Pi[2n]} \ket{\pi} = d_\R[n] \int_{S^{d-1}_\R} \ket{v}^{\otimes 2n} \text{d}v,
\]
with 
\[
d_\R[n] = \frac{2^{2n}n!~\Gamma(n+d/2)}{(2n)!~\Gamma(d/2)}.
\]
\label{lem:SphHilbRealLA}
\end{lemma}

We shall prove the lemma in its polynomial form. Therefore, note that $\{ \ket{x}^{\otimes n}~:~ x\in \R^d\}$ spans $\vee^n\R^d$, and that the following lemma is equivalent to the previous. 

\begin{lemma}[Real Spherical Hilbert identity -- Polynomial form]
For any $x\in \R^d$ we have 
\[
\|x\|^{2n} = d_\R[n] \int_{\|v\|=1} (\braket{x}{v})^{2n} \text{d}v,
\]
with 
\[
d_\R[n] = \frac{2^{2n}n!~\Gamma(n+d/2)}{(2n)!~\Gamma(d/2)}.
\]
\label{lem:SphHilbReal}
\end{lemma}

\begin{proof}
Using Lemma \ref{lem:GaussianCombRe} we find that 
\[
\|x\|^{2n} = \braket{x^{\otimes 2n}}{\Omega_d^{\otimes n}} = \frac{1}{|\Pi[2n]|} \sum_{\pi \in \Pi[2n]} \braket{x^{\otimes 2n}}{\pi}= \frac{1}{|\Pi[2n]|}\mathbb{E}[(\braket{x}{\hat{\phi}})^{2n}].
\]
To evaluate the expectation value in the last expression we shall use polar coordinates. For this we decompose $\ket{\hat{\phi}} = r\ket{v}$ with radial part 
\[
r = \sqrt{\braket{\hat{\phi}}{\hat{\phi}}}
\]
and spherical part $v:\Omega\to S^{d-1}$ distributed uniformly on the unit sphere. Then it is easy to see that we get
\[
\|x\|^{2n} = \frac{1}{|\Pi[2n]|}\mathbb{E}[r^{2n}]\frac{2\pi^{d/2}}{\Gamma(d/2)}\int_{\|v\|=1} (\braket{x}{v})^{2n} \text{d}v.
\]
Since $r$ is given as the length of the Gaussian vector $\ket{\hat{\phi}}$ we find that
\[
\mathbb{E}[r^{2n}] = \frac{1}{(2\pi)^{d/2}}\int^\infty_{0} r^{2n+d-1} e^{-\frac{1}{2}r^2} \text{d}r = \frac{2^{n}}{2\pi^{d/2}}\Gamma(n+d/2).
\] 
Finally this gives
\[
\|x\|^{2n} = \frac{2^{2n}n!~\Gamma(n+d/2)}{(2n)!~\Gamma(d/2)}\int_{\|v\|=1} (\braket{x}{v})^{2n} \text{d}v,
\]
where we used that $|\Pi[2n]| = (2n-1)(2n-3) \cdots 5 \cdot 3 \cdot 1 = \frac{(2n)!}{n! 2^n}$.

\end{proof}

\subsection{The complex case} \label{sec:app-HI-C}

Analogue to the real case we can derive Hilbert identities from Gaussian integrals. This is well-known within the community of quantum information theory (see \cite{harrow2013church}).

For any permutation $\sigma\in \mathfrak S_n$ we denote by $P_\sigma\in \mathcal{U}((\C^d)^{\otimes n})$ unitary operators defined by 
\[
P_\sigma\ket{i_1i_2\ldots i_n} = \ket{i_{\sigma(1)}i_{\sigma(2)}\ldots i_{\sigma(n)}},
\]
on the computational basis states. The projector onto the symmetric subspace $\vee^n \C^d\subset (\C^d)^{\otimes n}$ is given by
\[
P^n_\text{sym} = \frac{1}{n!} \sum_{\sigma\in \mathfrak S_n} P_\sigma.
\]

\begin{lemma}[Combinatorial formula for Gaussian integrals -- Complex case]
Let $\ket{\hat{\phi}}:\Omega\to \C^d$ be a vector-valued random variable where the entries $\hat{\phi}_i$ are iid complex Gaussian distributed with mean $0$ and variance $\mathbbm{E}(|\hat{\phi}_i|^2) = 1$. Then we have 
\[
\mathbb{E}[\ketbra{\hat{\phi}}{\hat{\phi}}^{\otimes n}] = \sum_{\sigma\in \mathfrak S_n} P_{\sigma} = n! P^n_\text{sym}.
\] 
\label{lem:GaussianCombComp}
\end{lemma}

\begin{proof}
We find that in the computational basis
\[
\bra{i_1i_2\ldots i_{n}}\mathbb{E}[\ketbra{\hat{\phi}}{\hat{\phi}}^{\otimes n}]\ket{i_{n+1}i_{n+2}\ldots i_{2n}} = \mathbb{E}[\hat{\phi}_{i_1}\ldots \hat{\phi}_{i_{n}}\overline{\hat{\phi}_{i_{n+1}}}\ldots \overline{\hat{\phi}_{i_{2n}}}]. 
\]
Denoting $\phi_k = Z^{(0)}_k + iZ^{(1)}_k$ and expanding the previous expectation value we find
\begin{align*}
\mathbb{E}[\hat{\phi}_{i_1}\ldots &\hat{\phi}_{i_{n}}\overline{\hat{\phi}_{i_{n+1}}}\ldots \overline{\hat{\phi}_{i_{2n}}}] = \sum_{t_1,\ldots, t_{2n}\in \{ 0,1\}} i^{t_1}\cdots i^{t_n}(-i)^{t_{n+1}}\cdots (-i)^{t_{2n}} \mathbb{E}[Z^{(t_1)}_{i_1}\cdots Z^{(t_{2n})}_{i_{2n}}] \\
&= \sum_{t_1,\ldots, t_{2n}\in \{ 0,1\}} i^{t_1}\cdots i^{t_n}(-i)^{t_{n+1}}\cdots (-i)^{t_{2n}} \sum_{\pi\in \Pi[2n]}\prod_{(r,s)\in \pi} \mathbb{E}[Z^{(t_r)}_{i_r} Z^{(t_{s})}_{i_{s}}] \\
& = \sum_{\pi\in \Pi[2n]}\quad\lbr\sum_{t_1,\ldots, t_{2n}\in \{ 0,1\}} i^{t_1}\cdots i^{t_n}(-i)^{t_{n+1}}\cdots (-i)^{t_{2n}} \prod_{(r,s)\in \pi} \delta_{t_r t_s}\mathbb{E}[Z^{(t_r)}_{i_r} Z^{(t_{s})}_{i_{s}}]\rbr
\end{align*}
where we used Wick's formula \cite{isserlis1918formula} in the second equality and we write $(r,s)\in \pi$ for the transpositions $\tau_i = (r,s)$ such that $\pi = \tau_1\circ \cdots \circ \tau_n$. In the last step we used that real and imaginary parts of the $\hat{\phi}_{i}$ are independent (Gaussian) random variables of mean $0$. Now observe that only $\pi\in \Pi[2n]$ pairing elements from $\{ 1,\ldots n\}$ with elements from $\{ n+1 ,\ldots 2n\}$ contribute to the above sum. For all other pairings the second sum (i.e.~the expression in the square brackets) vanishes due to cancellations. Therefore we have  
\begin{align*}
&\sum_{\pi\in \Pi[2n]}\quad \sum_{t_1,\ldots, t_{2n}\in \{ 0,1\}} i^{t_1}\cdots i^{t_n}(-i)^{t_{n+1}}\cdots (-i)^{t_{2n}} \prod_{(r,s)\in \pi} \delta_{t_r t_s}\mathbb{E}[Z^{(t_r)}_{i_r} Z^{(t_{s})}_{i_{s}}] \\
&= \sum_{\sigma\in \mathfrak S_n}\quad\sum_{t_1,\ldots, t_{n}\in \{ 0,1\}} \prod^n_{k=1} \mathbb{E}[Z^{(t_k)}_{i_k} Z^{(t_k)}_{i_{n+\sigma(k)}}]\\
& = \sum_{\sigma\in \mathfrak S_n}\prod^n_{k=1} \delta_{i_k i_{n+\sigma(k)}}
\end{align*}
where we replaced $\pi\in \Pi[2n]$ pairing elements from $\{ 1,\ldots n\}$ with elements from $\{ n+1 ,\ldots 2n\}$ by permutations $\sigma\in \mathfrak S_n$ by setting $\sigma(r)=s$ whenever $(r,n+s)\in \pi$ for $r\in \{ 1,\ldots n\}$. Finally, note that 
\[
\bra{i_1i_2\ldots i_{n}}\mathbb{E}[\ketbra{\hat{\phi}}{\hat{\phi}}^{\otimes n}]\ket{i_{n+1}i_{n+2}\ldots i_{2n}} = \sum_{\sigma\in \mathfrak S_n}\prod^n_{k=1} \delta_{i_k i_{n+\sigma(k)}} = \sum_{\sigma\in \mathfrak S_n}\bra{i_1i_2\ldots i_{n}}P_\sigma\ket{i_{n+1}i_{n+2}\ldots i_{2n}},
\]
which finishes the proof.
\end{proof}

In the following let $\mathcal S^{d-1} \subseteq \C^d$ denote the complex unit sphere. 

\begin{lemma}[Complex Spherical Hilbert identity--Linear algebra form]
We have
\[
P^{(n)}_{sym} = d[n] \int_{\mathcal S^{d-1}} \ketbra{\phi}{\phi}^{\otimes n} \mathrm{d}\phi,
\]
where $d[n] = \dim \vee^n\C^d =  \binom{n+d-1}{n}$.
\label{lem:SphHilbCompLA}
\end{lemma}

Again we shall prove the lemma in its polynomial form. Therefore, note that $\{ \ket{x}^{\otimes n}~:~ x\in \C^d\}$ spans $\vee^n\C^d$, and that the following lemma is equivalent to the previous.

\begin{lemma}[Complex Spherical Hilbert identity--Polynomial form]
For any $x\in \C^d$ we have 
\[
\|x\|^{2n} = d[n] \int_{\|v\|=1} |\braket{x}{v}|^{2n} \text{d}v.
\]

\label{lem:SphHilbComp}
\end{lemma}

\begin{proof}
Using Lemma \ref{lem:GaussianCombComp} we find that 
\[
\|x\|^{2n} = \braket{x^{\otimes n}}{x^{\otimes n}} = \bra{x^{\otimes n}}P^{(n)}_\text{sym}\ket{x^{\otimes n}}= \frac{1}{n!}\mathbb{E}[|\braket{x}{\hat{\phi}}|^{2n}].
\]
To evaluate the expectation value in the last expression we use again polar coordinates. For this we decompose $\ket{\hat{\phi}} = r\ket{v}$ with radial part 
\[
r = \sqrt{\braket{\hat{\phi}}{\hat{\phi}}}
\]
and spherical part $v:\Omega\to S^{d-1}$ distributed uniformly on the (complex) unit sphere. Then it is easy to see that we get
\[
\|x\|^n = \frac{1}{n!}\mathbb{E}[r^{2n}]\frac{2\pi^d}{(d-1)!}\int_{\|v\|=1} \braket{x}{v}^{2n} \text{d}v.
\]
Since $r$ is given as the length of the Gaussian vector $\ket{\hat{\phi}}$ we find that 
\[
\mathbb{E}[r^{2n}] = \frac{1}{2^n(2\pi)^{d}}\int^\infty_{0} r^{2n+2d-1} e^{-\frac{1}{2}r^2} \text{d}r = \frac{1}{2\pi^{d}}\Gamma(n+d) = \frac{1}{2\pi^{d}}(n+d-1)!.
\] 
Here the additional factor $1/2^n$ compared to the real case comes from the normalization of the Gaussian random variables appearing in the real and imaginary parts of the entries of $\ket{\hat{\phi}}$.
Finally this gives
\[
\|x\|^n = \frac{(n+d-1)!}{n!(d-1)!}\int_{\|v\|=1} \braket{x}{v}^{2n} \text{d}v.
\]

\end{proof}

\section{Simple complex spherical designs}\label{spherical-design-appendix}

In this appendix we present a simple method to construct complex spherical designs that is inspired and closely related to a method by Hausdorff \cite{hausdorff1909hilbertschen} for the real case. We do not claim, that this method is new. In fact similar methods can be found in the literature \cite{nesterenko2006waring,scott2006tight}, but we have not found a truly elementary account in the complex case. Let us start with the definition of a complex spherical design (note that in the literature, the objects introduced below are sometimes called \emph{weighted} complex spherical designs).

\begin{definition}[Complex spherical design]
	For $N,d,n\in\N$ a complex spherical $n$-design of order $N$ on $\C^d$ is a set of vectors $\{\ket{\gamma_i}\}^{N}_{i=1}\subset \C^d$ and a set of weights $\{ p_i\}^N_{i=1}\subset \R^+$ such that
	\begin{equation}
	P^{(n)}_{\text{sym}} = \sum^N_{i=1} p_i \ketbra{\gamma_i}{\gamma_i}^{\otimes n}. 
	\label{equ:condSphDe}
	\end{equation}
	\label{def:complSphDe}
\end{definition}

For our construction we need the well-known family of orthogonal Laguerre polynomials defined as
\[
L_m(x) = \frac{e^x}{m!}\frac{d^m}{dx^m}( e^{-x}x^m),
\]  
for any $m\in\N$. The following theorem summarizes some well-known properties of these polynomials and we refer to \cite{szeg1939orthogonal} for more details.

\begin{theorem}[Properties of Laguerre polynomials]\hfill
\begin{enumerate}
\item For any $m\in\N$ the Laguerre polynomial $L_m$ is of degree $m$ and has $m$ distinct zeros in $(0,\infty)$. 
\item For any $m\in\N$ we have 
\[
\int^\infty_0 L_m(x)p(x)e^{-x} dx = 0
\]
for any polyomial $p(x)$ of degree $\text{deg}(p)\leq m-1$.
\end{enumerate}
\label{thm:propLaguerre}
\end{theorem}

Note that condition \eqref{equ:condSphDe} from Definition \ref{def:complSphDe} is equivalent to the polynomial identity
\begin{equation}
( |y_1|^2 + \ldots + |y_d|^2)^n = \sum^N_{i=1} p_i ( \gamma_i(1) \overline{y_1} + \ldots + \gamma_i(d) \overline{y_d})^n ( \overline{\gamma_i(1)} y_1 + \ldots + \overline{\gamma_i(d)} y_d)^n
\label{equ:complexHilbert}
\end{equation}
for any $d$-tuple of complex numbers $y_1,\ldots ,y_d\in\C$. Here $\gamma_i(k)$ denotes the $k$th entry of the vector $\ket{\gamma_i}\in\C^d$ from Definition \ref{def:complSphDe} in the computational basis. Note that this identity is the complex analogue of a Hilbert identity. To find coefficients $\gamma_i(j)\in\C$ and weights $p_i\in\R^+$ satisfying this identity we need the following lemma.

\begin{lemma}[Moment problem]
For any $m\in\N$ there exist weights $w_j\geq 0$ and $\alpha_j\in\C$ for any $j\in\{ 1,\ldots ,m^2\}$ such that for $k,l\in\{ 0,\ldots , m-1\}$ we have 
\[
\sum^{m^2}_{j=1} w_j \overline{\alpha_j}^k\alpha^l_j = \begin{cases} k!, \text{ for }k=l\\
0, \text{ else.}\end{cases}
\]
\label{lem:Moment}
\end{lemma} 
\begin{proof}
Fix $m\in\N$. For $s\in\{ 1,\ldots ,m\}$ let $\beta_s\in (0,\infty)$ denote the zeros of the Laguerre polynomial $L_m$. According to Theorem \ref{thm:propLaguerre} these are distinct positive numbers. Now consider for any $s,t\in\{ 1,\ldots ,m\}$ the complex numbers
\[
\alpha_{st} = \sqrt{\beta_s}e^{\frac{2\pi i}{m}t}.
\] 
For arbitrary real numbers $w_{st} = w_s/m$ and any $k,l\in\{ 0,\ldots ,m-1\}$ we compute 
\[
\sum^m_{s=1}\sum^m_{t=1} w_{st}\overline{\alpha_{st}}^k\alpha_{st}^l = \begin{cases} \sum^m_{s=1}w_s \beta^{k}_s, \text{ for }k=l \\
0, \text{ else.}\end{cases} 
\]
Now note that the set of equations 
\[
\sum^m_{s=1}w_s \beta^{k}_s = k!, \text{ for }k\in\{ 0,\ldots ,m-1\}
\]
in unknowns $w_s$ is a Vandermonde system. As the $\beta_s$ are distinct by definition (as zeros of $L_m$) it has a unique solution $w_s\in\R$ for $s\in\{ 1,\ldots ,m\}$. It remains to show that $w_s\geq 0$.

Note that by the above for any polynomial $q(x) = \sum^{m-1}_{k=0}a_k x^k$ of $\text{deg}( q)\leq m-1$ we have 
\begin{equation}
\int^{\infty}_0 q(x)e^{-x}dx = \sum^{m-1}_{k=0} a_k k! = \sum^{m-1}_{k=0} a_k \sum^{m}_{s=1}w_s \beta^k_s = \sum^m_{s=1} w_s q(\beta_s)
\label{equ:quadq}
\end{equation}
where we used the elementary integral
\[
\int^\infty_0 x^k e^{-x} dx = k!
\]
for any $k\in\N$. By polynomial long division we can write any polynomial $Q(x)$ of degree $\text{deg}( Q)\leq 2m -1$ as
\[
Q(x) = L_{m}(x)g_{m-1}(x) + q(x).
\]
for polynomials $g_{m-1}(x), q(x)$ of degrees $\text{deg}(g_{m-1})\leq m-1$ and $\text{deg}(q)\leq m-1$. Now observe that 
\begin{equation}
\int^\infty_0 Q(x)e^{-x} dx = \int^\infty_0 q(x)e^{-x} dx = \sum^m_{s=1} w_s q(\beta_s) = \sum^m_{s=1} w_s Q(\beta_s)
\label{equ:quadQ}
\end{equation}
where we used that 
\[
\int^\infty_0 L_{m}(x)g_{m-1}(x)e^{-x} dx = 0
\]
by Theorem \ref{thm:propLaguerre} for the first equality, \eqref{equ:quadq} for the second equality, and for the third equality that by definition $L_m(\beta_s) = 0$ for all $s\in\{ 1,\ldots ,m\}$. For any $l\in\{ 1,\ldots ,m\}$ consider the polynomial 
\[
Q_l(x) = \lb \frac{L_m(x)}{x-\beta_l}\rb^2.
\]
Note that $\text{deg}(Q_l)= 2m-2$ and that $Q_l(x)\geq 0$ for any $x\in\R$. Using \eqref{equ:quadQ} it follows that 
\[
w_l (L'_m(\beta_l))^2 = \sum^m_{s=1} w_s Q_l(\beta_s) = \int^\infty_0 Q_l(x)e^{-x} dx \geq 0.
\] 
As $(L'_m(\beta_l))^2 >0$ (since $L_m$ has no degenerate zeros) it follows that $w_l\geq 0$ for any $l\in\{ 1,\ldots ,m\}$.
\end{proof}

Using the previous lemma we can explicitly construct a complex spherical design:

\begin{theorem}[Complex spherical design]
For any $n,d\in\N$ set $N = (n+1)^{2d}$. Then there exist a set of vectors $\{\ket{\gamma_i}\}^{N}_{i=1}\subset \C^d$ and a set of weights $\{ p_i\}^N_{i=1}\subset \R^+$ satisfying \eqref{equ:condSphDe}. 
\end{theorem}
\begin{proof}
Consider the weights $w_j\geq 0$ and complex numbers $\alpha_j\in\C$ for any $j\in\{ 1,\ldots , (m+1)^2\}$ constructed in Lemma \ref{lem:Moment} such that 
\[
\sum^{m^2}_{j=1} w_j \overline{\alpha_j}^k\alpha^l_j = \begin{cases} k!, \text{ for }k=l\\
0, \text{ else.}\end{cases}
\]
holds for all $k,l\in\{ 1,\ldots ,m\}$. Now we set
\begin{align*}
p_{i_1\ldots i_d} &= \frac{1}{n!}w_{i_1}\cdots w_{i_d} \geq 0 \\
\ket{\gamma_{i_1\cdots i_d}} &= (\alpha_{i_1},\ldots ,\alpha_{i_d})^T \in \C^{d}
\end{align*}
for any $d$-tuple $i_1,\ldots ,i_d\in\{ 1,\ldots , m+1\}$. Finally we can compute
\begin{align*}
\sum_{i_1,i_2,\ldots,i_d} &p_{i_1\cdots i_d} ( \gamma_{i_1\cdots i_d}(1) \overline{y_1} + \ldots + \gamma_{i_1\cdots i_d}(d) \overline{y_d})^n ( \overline{\gamma_{i_1\cdots i_d}(1)} y_1 + \ldots + \overline{\gamma_{i_1\cdots i_d}(d)} y_d)^n \\
&=\frac{1}{n!}\sum_{i_1,i_2,\ldots,i_d} w_{i_1}\cdots w_{i_d} ( \alpha_{i_1}(1) \overline{y_1} + \ldots + \alpha_{i_d}(d) \overline{y_d})^n ( \overline{\alpha_{i_1}(1)} y_1 + \ldots + \overline{\alpha_{i_d}(d)} y_d)^n \\
&=\sum_{i_1,i_2,\ldots,i_d} w_{i_1}\cdots w_{i_d}\sum_{{\substack{k_1+\cdots k_d = n \\ l_1+\cdots l_d = n}}} \frac{n!}{k_1!l_1!\cdots k_d!l_d!} \overline{\alpha_{i_1}}^{k_1}\alpha^{l_1}_{i_1}\ldots\overline{\alpha_{i_d}}^{k_d}\alpha^{l_d}_{i_d}y_1^{k_1}\overline{y_1}^{l_1}\cdots y_d^{k_d}\overline{y_d}^{l_d}  \\
& = \sum_{{\substack{k_1+\cdots k_d = n \\ l_1+\cdots l_d = n}}} \frac{n!}{l_1!k_1!\cdots l_d!k_d!} ( \sum_{i_1} w_{i_1}\overline{\alpha_{i_1}}^{k_1}\alpha^{l_1}_{i_1})\cdots ( \sum_{i_d} w_{i_d}\overline{\alpha_{i_d}}^{k_d}\alpha^{l_d}_{i_d}) y_1^{k_1}\overline{y_1}^{l_1}\cdots y_d^{k_d}\overline{y_d}^{l_d} \\
&= \sum_{k_1+\cdots k_d=n} \frac{n!}{k_1!\cdots k_d!} |y_1|^{2k_1}\cdots |y_d|^{2k_d}\\
& = ( |y_1|^2 + \cdots |y_d|^2)^n. 
\end{align*}
The previous computation verifies \eqref{equ:complexHilbert} for the vectors and weights constructed previously. Therefore, \eqref{equ:condSphDe} holds. 
\end{proof}

\bibliographystyle{quantum}
\bibliography{rag-qit}

\begin{thebibliography}{10}

\bibitem{bochnak2013real}
Jacek Bochnak, Michel Coste, and Marie-Françoise Roy.
\newblock ``Real algebraic geometry''.
\newblock \href{https://dx.doi.org/10.1007/978-3-662-03718-8}{Volume~36}.
\newblock Springer Science and Business Media. ~(2013).

\bibitem{marshall2008positive}
Murray Marshall.
\newblock ``Positive polynomials and sums of squares''.
\newblock \href{https://dx.doi.org/10.1090/surv/146}{Volume 146}.
\newblock American Mathematical Soc. ~(2008).

\bibitem{krivine1964anneaux}
Jean-Louis Krivine.
\newblock ``Anneaux préordonnés''.
\newblock \href{https://dx.doi.org/10.1007/bf02807438}{Journal d'analyse
  mathématique {\bf 12}, 307--326}~(1964).

\bibitem{stengle1974nullstellensatz}
Gilbert Stengle.
\newblock ``A nullstellensatz and a positivstellensatz in semialgebraic
  geometry''.
\newblock \href{https://dx.doi.org/10.1007/bf01362149}{Mathematische Annalen
  {\bf 207}, 87--97}~(1974).

\bibitem{schmuedgen1991positivstellensatz}
Konrad Schmüdgen.
\newblock ``The k-moment problem for compact semi-algebraic sets''.
\newblock \href{https://dx.doi.org/10.1007/BF01446568}{Mathematische Annalen
  {\bf 289}, 203--206}~(1991).

\bibitem{putinar1993positive}
Mihai Putinar.
\newblock ``Positive polynomials on compact semi-algebraic sets''.
\newblock \href{https://dx.doi.org/10.1512/iumj.1993.42.42045}{Indiana
  University Mathematics Journal {\bf 42}, 969--984}~(1993).

\bibitem{artin1927zerlegung}
Emil Artin.
\newblock ``Über die zerlegung definiter funktionen in quadrate''.
\newblock In Abhandlungen aus dem mathematischen Seminar der Universität
  Hamburg.
\newblock \href{https://dx.doi.org/10.1007/bf02952513}{Volume~5, pages
  100--115}.
\newblock SpringerSpringer Science and Business Media LLC~(1927).

\bibitem{reznick1995uniform}
Bruce Reznick.
\newblock ``Uniform denominators in hilbert's seventeenth problem''.
\newblock \href{https://dx.doi.org/10.1007/BF02572604}{Math. Z. {\bf 220},
  75--97}~(1995).

\bibitem{to2006effective}
Wing-Keung To and Sai-Kee Yeung.
\newblock ``Effective isometric embeddings for certain hermitian holomorphic
  line bundles''.
\newblock \href{https://dx.doi.org/10.1112/S0024610706022708}{Journal of the
  London Mathematical Society {\bf 73}, 607--624}~(2006).

\bibitem{harrow2013church}
Aram~W Harrow.
\newblock ``The church of the symmetric subspace''~(2013).

\bibitem{chiribella2010quantum}
Giulio Chiribella.
\newblock ``On quantum estimation, quantum cloning and finite quantum de
  finetti theorems''.
\newblock In Wim van Dam, Vivien~M. Kendon, and Simone Severini, editors,
  Conference on Quantum Computation, Communication, and Cryptography.
\newblock \href{https://dx.doi.org/10.1007/978-3-642-18073-6_2}{Volume 6519,
  pages 9--25}.
\newblock SpringerSpringer Berlin Heidelberg~(2010).

\bibitem{fang2020sum}
Kun Fang and Hamza Fawzi.
\newblock ``The sum-of-squares hierarchy on the sphere and applications in
  quantum information theory''.
\newblock \href{https://dx.doi.org/10.1007/s10107-020-01537-7}{Mathematical
  Programming {\bf 190}, 1--30}~(2020).

\bibitem{faybusovich2004global}
Leonid Faybusovich.
\newblock ``Global optimization of homogeneous polynomials on the simplex and
  on the sphere''.
\newblock In Frontiers in global optimization.
\newblock \href{https://dx.doi.org/10.1007/978-1-4613-0251-3_6}{Pages
  109--121}.
\newblock Springer~(2004).

\bibitem{doherty2012convergence}
Andrew~C Doherty and Stephanie Wehner.
\newblock ``Convergence of sdp hierarchies for polynomial optimization on the
  hypersphere''~(2012).
\newblock  url:~\url{http://arxiv.org/abs/1210.5048}.

\bibitem{bhatia1997matrix}
Rajendra Bhatia.
\newblock ``Matrix analysis''.
\newblock \href{https://dx.doi.org/10.1007/978-1-4612-0653-8}{Volume 169}.
\newblock Springer Science and Business Media. ~(1997).

\bibitem{husimi1940some}
Kôdi Husimi.
\newblock ``Some formal properties of the density matrix''.
\newblock \href{https://dx.doi.org/10.11429/ppmsj1919.22.4_264}{Proceedings of
  the Physico-Mathematical Society of Japan. 3rd Series {\bf 22},
  264--314}~(1940).

\bibitem{delsarte1991spherical}
Philippe Delsarte, Jean-Marie Goethals, and Johan~Jacob Seidel.
\newblock ``Spherical codes and designs''.
\newblock In Geometry and Combinatorics.
\newblock
  \href{https://dx.doi.org/10.1016/B978-0-12-189420-7.50013-X}{Volume~6, pages
  68--93}.
\newblock Elsevier~(1991).

\bibitem{scott2006tight}
Andrew~J Scott.
\newblock ``Tight informationally complete quantum measurements''.
\newblock \href{https://dx.doi.org/10.1088/0305-4470/39/43/009}{Journal of
  Physics A: Mathematical and General {\bf 39}, 13507}~(2006).

\bibitem{wilde2017quantum}
Mark~M Wilde.
\newblock ``Quantum information theory''.
\newblock \href{https://dx.doi.org/10.1017/cbo9781139525343}{Cambridge
  University Press}. ~(2017).

\bibitem{keyl1999optimal}
Michael Keyl and Reinhard~F Werner.
\newblock ``Optimal cloning of pure states, testing single clones''.
\newblock \href{https://dx.doi.org/10.1063/1.532887}{Journal of Mathematical
  Physics {\bf 40}, 3283--3299}~(1999).

\bibitem{graham1989concrete}
Ronald~L Graham, Donald~E Knuth, Oren Patashnik, and Stanley Liu.
\newblock ``Concrete mathematics: a foundation for computer science''.
\newblock \href{https://dx.doi.org/10.1063/1.4822863}{Computers in Physics {\bf
  3}, 106--107}~(1989).

\bibitem{hobson1931theory}
Ernest~W Hobson.
\newblock ``The theory of spherical and ellipsoidal harmonics''.
\newblock \href{https://dx.doi.org/10.2307/3607762}{Volume~40, page 599}.
\newblock CUP Archive. ~(1931).

\bibitem{quillen1968representation}
Daniel~G Quillen.
\newblock ``On the representation of hermitian forms as sums of squares''.
\newblock \href{https://dx.doi.org/10.1007/bf01389773}{Inventiones mathematicae
  {\bf 5}, 237--242}~(1968).

\bibitem{nam2008effective}
Mok~Hoi Nam.
\newblock ``Effective aspects of positive semi-definite real and complex
  polynomials''.
\newblock Master's thesis.
\newblock Department of Mathematics, National University of Singapore.
\newblock ~(2008).

\bibitem{renner2007symmetry}
Renato Renner.
\newblock ``Symmetry of large physical systems implies independence of
  subsystems''.
\newblock \href{https://dx.doi.org/10.1038/nphys684}{Nature Physics {\bf 3},
  645}~(2007).

\bibitem{koenig2009most}
Robert König and Graeme Mitchison.
\newblock ``A most compendious and facile quantum de finetti theorem''.
\newblock \href{https://dx.doi.org/10.1063/1.3049751}{Journal of Mathematical
  Physics {\bf 50}, 012105}~(2009).

\bibitem{hilbert1909beweis}
David Hilbert.
\newblock ``Beweis für die darstellbarkeit der ganzen zahlen durch eine feste
  anzahl n-ter potenzen (waringsches problem)''.
\newblock \href{https://dx.doi.org/10.1007/978-3-642-50831-8_11}{Mathematische
  Annalen {\bf 67}, 281--300}~(1909).

\bibitem{ellison1971waring}
William~J Ellison.
\newblock ``Waring's problem''.
\newblock \href{https://dx.doi.org/10.1080/00029890.1971.11992689}{The American
  Mathematical Monthly {\bf 78}, 10--36}~(1971).

\bibitem{isserlis1918formula}
Leon Isserlis.
\newblock ``On a formula for the product-moment coefficient of any order of a
  normal frequency distribution in any number of variables''.
\newblock \href{https://dx.doi.org/10.1093/biomet/12.1-2.134}{Biometrika {\bf
  12}, 134--139}~(1918).

\bibitem{hausdorff1909hilbertschen}
Felix Hausdorff.
\newblock ``Zur hilbertschen lösung des waringschen problems''.
\newblock \href{https://dx.doi.org/10.1007/bf01450406}{Mathematische Annalen
  {\bf 67}, 301--305}~(1909).

\bibitem{nesterenko2006waring}
Yu~V Nesterenko.
\newblock ``On waring's problem (elementary methods)''.
\newblock \href{https://dx.doi.org/10.1007/s10958-006-0266-8}{Journal of
  Mathematical Sciences {\bf 137}, 4699--4715}~(2006).

\bibitem{szeg1939orthogonal}
Gabor Szego.
\newblock ``Orthogonal polynomials''.
\newblock \href{https://dx.doi.org/10.2307/1989699}{Volume~23}.
\newblock American Mathematical Soc. ~(1939).

\end{thebibliography}

\end{document}